\documentclass[twocolumn,10pt]{IEEEtran}
\DeclareUnicodeCharacter{3000}{ }

\usepackage[bottom=0.9in,top=0.6in,left=0.625in,right=0.625in]{geometry}

\ifCLASSINFOpdf
  \usepackage[pdftex]{graphicx}
  \graphicspath{{../pdf/}{../jpeg/}}
  \DeclareGraphicsExtensions{.pdf,.jpeg,.png}
\else
  \usepackage[dvips]{graphicx}
  \graphicspath{{../eps/}}
  \DeclareGraphicsExtensions{.eps}
\fi

\usepackage{epstopdf}
\usepackage{multirow}
\usepackage{float}
\usepackage[cmex10]{amsmath}
\usepackage {amssymb}
\usepackage{graphicx}
\usepackage{array}
\usepackage{mdwmath}
\usepackage{mdwtab}
\usepackage{eqparbox}

\usepackage{nomencl}
\usepackage{cite}
\usepackage{algorithm}
\usepackage{algorithmic}
\usepackage{mathtools}
\usepackage{makeidx}
\usepackage{ifthen}
\usepackage{subfigure}
\usepackage{caption}
\usepackage{color}
\usepackage{multirow,array}
\usepackage{pdflscape}
\usepackage{bm,upgreek}

\DeclarePairedDelimiter\floor{\lfloor}{\rfloor}
\makenomenclature

\usepackage{epstopdf}
\usepackage{multirow}
\usepackage{float}
\usepackage[cmex10]{amsmath}
\usepackage {amssymb}
\usepackage{graphicx}
\usepackage{array}
\usepackage{mdwmath}
\usepackage{mdwtab}
\usepackage{eqparbox}
\usepackage{nomencl}
\usepackage{cite}
\usepackage{algorithm}
\usepackage{algorithmic}
\usepackage{makeidx}
\usepackage{ifthen}
\usepackage{subfigure}
\usepackage{caption}
\usepackage{amsthm}

\renewcommand{\nomgroup}[1]{%
\ifthenelse{\equal{#1}{I}}{\item[\textbf{Indices}]}{%
\ifthenelse{\equal{#1}{A}}{\item[\textbf{Abbreviations}]}{%
\ifthenelse{\equal{#1}{V}}{\item[\textbf{Variables}]}{%
\ifthenelse{\equal{#1}{P}}{\item[\textbf{Parameters and Constants}]}{%
}
}
}
}
}

\ifCLASSINFOpdf
  \usepackage[pdftex]{graphicx}
  \graphicspath{{../pdf/}{../jpeg/}}
  \DeclareGraphicsExtensions{.pdf,.jpeg,.png}
\else
  \usepackage[dvips]{graphicx}
  \graphicspath{{../eps/}}
  \DeclareGraphicsExtensions{}
\fi
\usepackage{epstopdf}

\usepackage{multirow,tikz}
\usepackage{xspace,amsmath,amsfonts,amssymb,tikz,multirow,float,newclude,pgfplots, stmaryrd,lipsum,mathtools, tikz, amsthm,changepage }
\usepackage{xspace,amsmath,amsfonts,amssymb,tikz,multirow, pgfplots,pbox}
\usepackage{caption}
\usepackage{xspace}
\usepackage{latexsym}
\usepackage{xcolor}
\usepackage{graphicx}
\usepackage{enumitem}
\usepackage{comment,verbatim}
\usepackage{array}
\usepackage{booktabs}
\usepackage{cite}
\usepackage{fancyhdr}

\usepackage{comment}
\usepackage{epstopdf}
\usepackage{float}

\newcommand{\beq}{\begin{equation}}
\newcommand{\eeq}{\end{equation}}
\newcommand{\beqn}{\begin{eqnarray}}
\newcommand{\eeqn}{\end{eqnarray}}
\newcommand{\beqno}{\begin{eqnarray*}}
\newcommand{\eeqno}{\end{eqnarray*}}
\newcommand{\bma}{\begin{displaymath}}
\newcommand{\ema}{\end{displaymath}}
\newcommand{\bnu}{\begin{enumerate}}
\newcommand{\enu}{\end{enumerate}}
\newcommand{\bce}{\begin{center}}
\newcommand{\ece}{\end{center}}
\newcommand{\btb}{\begin{tabular}}
\newcommand{\etb}{\end{tabular}}

\def\bw{{\bm{w}}}
\def\bx{{\bm{x}}}
\def\bt{{\bm{t}}}
\def\bz{{\bm{z}}}
\def\bv{{\bm{v}}}

\def\bb{{\bm{b}}}
\def\bd{{\bm{d}}}
\def\bu{{\bm{u}}}

\newcommand{\one}{{\mathbf{1}}}

\def\revtwo#1{{\color{black}#1}}

\newtheorem{theorem}{Theorem}[section]
\newtheorem{lemma}[theorem]{Lemma}
\newtheorem{remark}[theorem]{Remark}

\allowdisplaybreaks[2]

\begin{document}

\title{Delay-Aware Robust Edge Network Hardening Under Decision-Dependent Uncertainty}

\author{\IEEEauthorblockN{Jiaming Cheng,~\IEEEmembership{Student Member,~IEEE}, Duong~Thuy~Anh~Nguyen,~\IEEEmembership{Student Member,~IEEE},  \\
Ni Trieu~\IEEEmembership{Member,~IEEE}, and Duong Tung Nguyen,~\IEEEmembership{Member,~IEEE}}  \vspace{-0.4cm}
\thanks{The first two authors have equal contributions. The authors are with the  Ira A. Fulton Schools of Engineering, Arizona State University, Tempe, AZ, USA. Email: \textit\{jiaming,dtnguy52,nitrieu,duongnt\}@asu.edu. } 
}
\maketitle

\begin{abstract}
Edge computing promises to offer low-latency and ubiquitous computation to numerous devices at the network edge. For delay-sensitive applications, link delays significantly affect service quality. These delays can fluctuate substantially over time due to various factors such as network congestion, changing traffic conditions, cyberattacks, component failures, and natural disasters. Thus, it is crucial to efficiently harden the edge network to mitigate link delay variation and ensure a stable and improved user experience. To this end, we propose a novel robust model for optimal edge network hardening, considering link delay uncertainty. Unlike existing literature that treats uncertainties as exogenous, our model incorporates an endogenous uncertainty set to properly capture the impact of hardening and workload allocation decisions on link delays. However, the endogenous set introduces additional complexity to the problem due to the interdependence between decisions and uncertainties. To address this, we present two efficient methods to transform the problem into a solvable form. Extensive numerical results demonstrate the effectiveness of the proposed approach in mitigating delay variations and enhancing system performance.
\end{abstract}

\begin{IEEEkeywords}
Edge computing,  delay reduction, network hardening, robust optimization, decision-dependent uncertainty. 
\end{IEEEkeywords}

\printnomenclature

\section{Introduction}
\label{intro}
The explosive growth of mobile devices and applications has triggered an 
unprecedented surge in mobile data traffic. As emerging services such as augmented/virtual reality (AR/VR), manufacturing automation, and autonomous driving continue to evolve, it is imperative to devise innovative solutions to effectively meet their stringent requirements. To this end, edge computing (EC) has been advocated as a vital computing paradigm that complements the traditional cloud to provide a superior user experience and enable various low-latency and ultra-reliable  Internet of Things (IoT) applications  \cite{wshi16}.

The performance and reliability of EC systems are susceptible to various uncertainties stemming from multiple sources, including extreme weather conditions, fluctuating resource demands, traffic spikes, user mobility, and changes in application performance and user behavior. Moreover, the escalating complexity and diversity of man-made attacks and cyber threats, such as malware attacks, cyberattacks, and insider threats, pose additional uncertainties and risks to EC systems. Consequently, designing EC systems with robust measures to manage system uncertainties as well as detect and mitigate potential threats and component failures is of utmost importance \cite{al2009comparative}. 

Effectively managing uncertainties is a key factor for achieving consistent performance, reliability, and an enhanced user experience in EC. A crucial but frequently overlooked aspect in the literature is link delay uncertainty, which may arise due to a variety of factors, including fluctuating workloads, network congestion, routing changes, link failures, and cyber-attacks. Specifically, workload is a major contributor to delay uncertainty, as higher traffic volumes on a link can increase packet delays. Network congestion is another critical factor, causing delays and packet loss, particularly during high-traffic periods.  Moreover, routing changes, prompting packets to follow different paths, can introduce additional delays and variability in packet transmission. In the event of link failures, congestion can arise, causing slower speeds as traffic is rerouted. Finally, deliberate network attacks targeting links can inundate the network with excessive traffic, thereby exacerbating delay issues.

Delay uncertainties can significantly affect the performance of edge network systems, especially for delay-sensitive applications. To address this issue, this paper focuses on enhancing the resilience of edge networks against delay uncertainties through various link hardening strategies, such as redundancy and failover, load balancing, traffic shaping, fault-tolerant routing, quality of service (QoS), and security mechanisms. Redundancy ensures that backup links and resources are available in the event of primary link failures \cite{Modiano14,fhe_tcc21}, while failover mechanisms automatically switch to backup links to guarantee network continuity \cite{failover16_ICC}. 
Upgrading hardware, increasing bandwidth, and investing in higher-quality fiber optic cables can also bolster the resilience of network links. Load balancing is another crucial strategy that evenly distributes network traffic across multiple links, preventing any single link from becoming overloaded and reducing the likelihood of delays and failures \cite{hsiao2001load}. 

Traffic shaping plays a vital role in controlling the traffic flow, avoiding link overload, and reducing delays and packet loss \cite{yu2017handling}. Fault-tolerant routing algorithms can help ensure that traffic continues to flow even when links fail, by rerouting traffic through alternate paths \cite{al2009comparative}. In addition, implementing QoS mechanisms and security measures is essential. QoS mechanisms prioritize traffic based on its importance, ensuring that critical traffic receives preferential treatment and minimizing delays and packet loss for critical applications. Security measures such as firewalls, intrusion detection systems, and other security mechanisms protect network links from attacks and unauthorized access, further fortifying the network \cite{jamming07,durkota2015optimal,dewri2007optimal}.

Different levels of hardening strategies yield varying degrees of network resilience, each with its distinct cost \cite{US_harden2}. Each level may contain one or several strategies, and the operator can flexibly choose which strategies to apply at each level for each link. Higher levels of hardening offer more robust protection, reducing the chance of network failures and improving overall network reliability \cite{US_harden1}. However, these advantages are accompanied by increased costs. The associated costs differ due to each hardening level potentially requiring distinct technologies, strategies, infrastructure enhancements, or ongoing maintenance efforts. For example, standard redundancy and failover mechanisms provide essential protection against link failures, but they may fall short for critical, delay-sensitive services. Advanced hardening strategies (e.g., high-grade encryption, fault-tolerance routing, and extensive traffic shaping) provide a higher level of security and resilience \cite{al2009comparative}. Implementing these measures often requires significant investments in software development and cybersecurity expertise \cite{Forbes,CISA}. Establishing redundant high-capacity links, like fiber optic cables, also requires considerable capital and ongoing operational costs.

By adopting these proactive strategies, edge networks can effectively mitigate the impact of delay uncertainties and enhance their overall resilience. However, hardening all network links is impractical and economically unfeasible. Therefore, it becomes crucial to determine which specific links should be hardened and at what level. Surprisingly, this critical problem has been largely overlooked in the existing literature. Moreover, the actual link delay and the level of uncertainty are directly  influenced by the workload and the chosen degree of link hardening. Conversely, the link hardening and workload allocation decisions are dependent on the link delays. This interdependence creates a fundamental and unresolved challenge in optimizing edge link hardening decisions.

This paper aims to bridge this gap by proposing a novel robust optimization (RO) framework for optimal link hardening in EC, considering the influence of the hardening and workload allocation decisions on link delays. Specifically, to accurately model the link delays, we introduce a decision-dependent uncertainty (DDU) set that explicitly captures the interdependence between the uncertainties and decisions. The incorporation of the DDU set into the robust model provides a more realistic representation of uncertainty during actual operation. Unlike the traditional RO method using a fixed decision-independent uncertainty (DIU) set, the DDU set is adjustable. This allows decision-makers to proactively control the form and level of uncertainty through their decisions, thereby mitigating the inherent conservatism of the robust solution. However, the introduction of interdependencies between uncertainties and decisions adds complexity to the problem, resulting in a large-scale non-linear optimization problem with numerous bilinear terms. To overcome this challenge, we present two efficient exact algorithms that can solve the underlying robust model under DDU. Our main goal is to quantify the benefit and importance of considering endogenous uncertainty. The proposed model is primarily designed for edge infrastructure providers that own and operate distributed edge data centers and network links. For example, telecom companies such as Verizon, T-Mobile, and AT\&T can benefit directly from its implementation.
Additionally, cloud and edge infrastructure providers (e.g., Amazon, Google, and Equinix), which manage extensive network infrastructures, can also leverage the model.

Our contributions can be summarized as follows:
\begin{itemize}[leftmargin=0pt]
    \item[] \textit{1) \textbf{Modeling:}} We introduce a novel robust model that addresses the optimal edge network link hardening and workload allocation problem, aiming to minimize costs while enhancing user experience by reducing link delay variation. Unlike traditional RO models, which typically consider exogenous uncertainty \cite{RObook}, our proposed model incorporates an endogenous uncertainty set to capture the interdependencies between uncertain link delays and the decisions related to hardening and workload allocation. To the best of our knowledge, this is the first robust link-hardening model in EC that explicitly accounts for such decision-dependent uncertainty.
    \item[] \textit{2) \textbf{Techniques:}} To solve the challenging robust problem under DDU, we deviate from traditional RO models and develop two efficient and exact reformulations, termed \textit{RDDU} and \textit{e-RDDU}. These reformulations are derived through a series of transformations that convert the original problem into a Mixed Integer Linear Programming (MILP) form, which can be solved efficiently using standard solvers such as Gurobi\footnote{https://www.gurobi.com/} and Mosek\footnote{https://www.mosek.com/}. Notably, \textit{e-RDDU} enhances upon \textit{RDDU} by leveraging the specific structure of the DDU set and employing a clever transformation technique that substantially reduces the problem's size and complexity. 
    \item[] \textit{3) \textbf{Numerical results:}} Extensive simulations demonstrate the efficiency of the proposed scheme compared to three benchmarks: a \textit{no-hardening} scheme, a \textit{random hardening} scheme, and a \textit{stochastic scheme under DDU}. Also, sensitivity analyses were conducted to evaluate the impact of important system parameters on system performance. 
\end{itemize}

This paper is structured as follows. Section ~\ref{system} outlines the system model and problem formulation. Section~\ref{sol} details the proposed solution approach. The numerical results are presented in Section~\ref{results}. Section \ref{related_work} discusses related work followed by conclusions in Section~\ref{conc}.

\section{System Model and Problem Formulation}
\label{system}
\revtwo{This section outlines the system model. 
We also define key parameters, decision variables, and constraints, and formulate the deterministic model, followed by two robust models that consider both exogenous and endogenous uncertainties.}

\subsection{System Model}
This paper studies an edge network hardening problem that involves a budget-constrained EC platform \revtwo{(i.e., an edge infrastructure provider)}, aiming to optimize link hardening decisions to improve user experience by reducing delays. The platform manages a set  $\mathcal{J}$ of $J$ heterogeneous edge nodes (ENs) and offers edge resources to users in different areas, each represented by an access point (AP). The set of APs is denoted by $\mathcal{I}$, and the number of APs is $I$.  Let $i$ and $j$ signify the AP index and EN index, respectively. We consider the graph $\mathbb{G}(\mathcal{V},\mathcal{E})$, where $\mathcal{V}$ is the set of nodes including $I$ APs and $J$ ENs. The set $\mathcal{E}=\{(i,j):i\in \mathcal{I},j\in \mathcal{J}\}$ represents the set of links connecting APs with  ENs. Each link $(i,j)$ represents the logical link connecting area $i$ with EN $j$, which may encompass multiple physical links.

The  EN size can vary significantly, and each EN may consist of one or several edge servers. 
For simplicity, we consider only computing resources. The resource capacity at EN $j$ is denoted by $C_j$. Define $\lambda_i$ as the resource demand in area $i$.
To reduce the network delay, demand in each area should be served by its closest EN. However, the capacity of each EN is limited. Thus, the platform needs to optimize the workload allocation decision, considering the edge resource capacity constraints, to ensure service quality while reducing costs. Let $x_{i,j}$ be the workload allocated from area $i$ to EN $j$. User requests from each area must be either served by some ENs or dropped. We denote the amount of unmet demand and the penalty for each unit of unmet demand in area $i$ by $w_i$ and $s_i$, respectively. 

The QoS for latency-sensitive applications is heavily affected by link delays. The delay in the logical link between area  $i$ and EN $j$, which may encompass multiple physical link segments connecting intermediate nodes along the path between them, is denoted by $d_{i,j}$. Link delay can fluctuate dramatically over time due to various factors such as varying traffic conditions, cyberattacks, network congestion, and link/router failures. %Additionally, we
We consider $d_{i,j}$ as a time-varying and uncertain parameter instead of a constant as commonly assumed in the literature. To provide a superior user experience, the platform must aim to reduce both link delay and delay variation. Developing efficient models for optimal link hardening to mitigate delay variation and improve system resilience is of paramount importance.

The cost of link hardening can vary depending on several factors, such as the characteristics and locations of the links, as well as the types and levels of hardening. We assume that each link has $R$ levels of hardening, denoted by $h_{i,j}^r$ for level-$r$ hardening cost of link $(i,j)$, which increases as the hardening level increases. A higher hardening level implies improved link performance and reduced delay variation. 
Note that each hardening level may contain several link hardening strategies.
We introduce a binary variable $t_{i,j}^{r}$, where $t_{i,j}^{r}$  equals $1$ if link $(i,j)$ if level-$r$ hardening is applied to link $(i,j)$. We also denote the incremental link hardening cost between two adjacent levels for each link by  $\Delta h_{i,j}^r$ (i.e., $\Delta h_{i,j}^r =  h^{r}_{i,j} - h^{r-1}_{i,j}$) and  the incremental impact factor (IF) between adjacent levels by $\Delta \gamma_{i,j}^r$ (i.e., $\Delta \gamma_{i,j}^r =  \gamma^{r}_{i,j} - \gamma^{r-1}_{i,j}$). Without loss of generality, we assume  $\Delta h_{i,j}^r =  \Delta h$ and  $\Delta \gamma_{i,j}^r =  \Delta \gamma,  \forall i,j,r$.  Let $B$ be the total hardening budget of the platform. The platform aims to identify a subset of critical links for efficient hardening within budget constraints. Table \ref{tab:notation} summarizes the main notations.

\begin{table}[t!]
\vspace{0.2cm}
\begin{tabular}{|ll|}
\hline
\multicolumn{1}{|l|}{\textbf{Notation}} & \textbf{Definition}                      \\ \hline
\multicolumn{2}{|c|}{\textbf{Sets and indices}}                                       \\ \hline
\multicolumn{1}{|l|}{EN, AP}            & Edge node, Access point                  \\ \hline
\multicolumn{1}{|l|}{$\mathcal{I}, I$}  & Set and number of areas (APs)            \\ \hline
\multicolumn{1}{|l|}{$\mathcal{J}, J$}  & Set and number of edge nodes (ENs)       \\ \hline
\multicolumn{2}{|c|}{\textbf{Parameters}}                                          \\ \hline
\multicolumn{1}{|l|}{$h_{i,j}^{r}$}     & Level $r$ link hardening cost for link $(i,j)$ \\ \hline
\multicolumn{1}{|l|}{$C_j$}             & Computing resource capacity at EN $j$    \\ \hline
\multicolumn{1}{|l|}{B}                 & Total link hardening budget              \\ \hline
\multicolumn{1}{|l|}{$s_i$}             & Unmet demand penalty in area $i$                    \\ \hline
\multicolumn{1}{|l|}{$d_{i,j}$}         & Network delay between AP $i$ and EN $j$  \\ \hline
\multicolumn{1}{|l|}{$\gamma_{i,j}^{r}$}  & Impact factor of level $r$ link hardening along link $(i,j)$ \\ \hline
\multicolumn{1}{|l|}{$u_{i,j}$}  & Impact factor of workload along link $(i,j)$ \\ \hline
\multicolumn{1}{|l|}{$\rho$}            & Delay penalty                            \\ \hline
\multicolumn{1}{|l|}{$\Delta \gamma$}   & Incremental link hardening impact \\ %between two adjacent levels    \\ 
\hline
\multicolumn{1}{|l|}{$\Delta h$}   & Incremental link hardening cost \\ %between two adjacent levels     \\ 
\hline
\multicolumn{1}{|l|}{$\lambda_i$}         & Resource demand in area $i$              \\ \hline
\multicolumn{1}{|l|}{$\Gamma_1$, $\Gamma_2$}            & Uncertain budget in $\mathcal{D}_{1}$, $\mathcal{D}_{2}$                 \\ \hline
\multicolumn{1}{|l|}{$\Psi$}            & Scaling factor for link hardening cost ($h_{i,j}^{r}$)     \\ \hline
\multicolumn{2}{|c|}{\textbf{Variables}}                                           \\ \hline
\multicolumn{1}{|l|}{$x_{i,j}$}         & Workload allocated from AP $i$ to EN $j$ \\ \hline
\multicolumn{1}{|l|}{$w_i$}             & Unmet demand in area $i$                 \\ \hline
\multicolumn{1}{|l|}{$t_{i,j}^{r}$} & $\{0,1\}$, 1 if link $(i,j)$ is chosen for level-$r$ hardening           \\ \hline
\end{tabular}
\caption{Notations}
\label{tab:notation}
\vspace{-0.5cm}
\end{table}

\revtwo{Fig.~\ref{fig:model} presents a toy example of the system model with $2$ AP and $2$ ENs. The platform aims to serve users' requests from both areas. For simplicity, link conditions are represented as tuples and visualized using different colors. In scenario (a), the delay for link $(1,1)$ ranges from $[5,15]$ ms under normal conditions, while link $(1,2)$ experiences higher variability, ranging from $[5,20]$ ms. Without hardening, $\frac{3}{4}$ of the workload from AP$1$ is allocated to EN$1$, while $\frac{1}{4}$ is assigned to EN$2$. The platform may priotize allocating workload to EN$1$ if capacity is allowed due to its lower network delay. In scenario (b), proactive hardening is applied to link $(1,2)$, reducing the maximum delay deviation from $20$ ms to $12$ ms, thereby improving QoS. Consequently, the platform can priortize workload allocation to EN$2$ due to its lower network latency. This example highlights the importance of proactive hardening in mitigating uncertainty, enhancing system performance, and reducing operational costs. }

\begin{figure}[t!]
\centering
\includegraphics[width=0.43\textwidth,height=0.15\textheight]{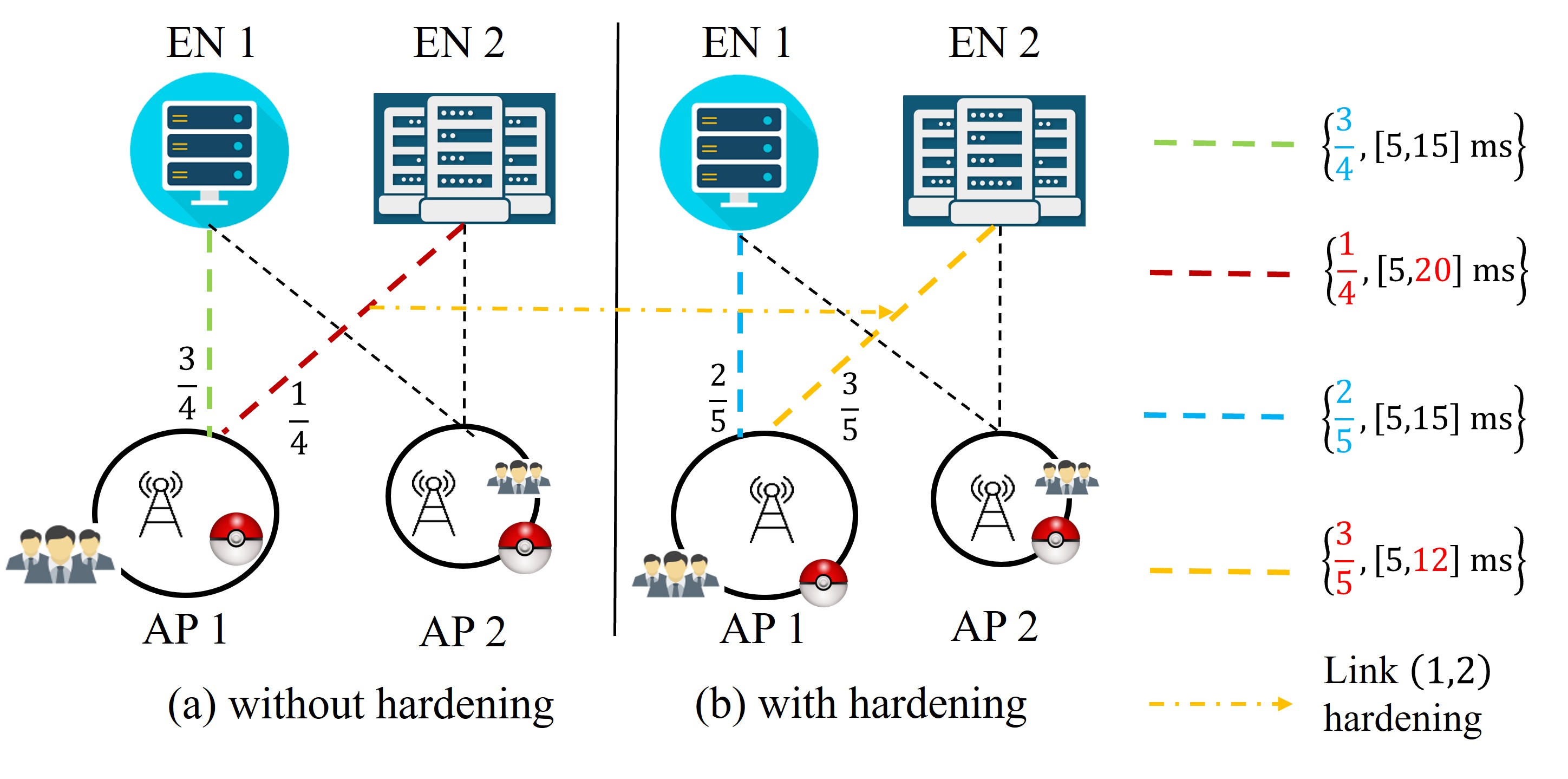}
\caption{\revtwo{System model}}
\label{fig:model}
\vspace{-0.5cm}
\end{figure}

\vspace{-0.5cm}
\subsection{Deterministic Hardening Model}
\label{deterministic}
\revtwo{We first formulate the deterministic problem, termed \textbf{DET}, for optimizing hardening and workload allocation decisions without considering system uncertainties.}
The EC platform aims to minimize the total cost of link hardening while improving user experience by reducing network delay and unmet demand. The \textbf{DET} problem can be expressed as follows:
\begin{subequations}
\label{DET}
\begin{align}
     \textbf{(DET)} \quad &\min_{\bt,\bx,\bw}  \:  \underbrace{\sum_{i,j} h_{i,j}^{r} t_{i,j}^{r}}_{\mathcal{C}_1(\bt)}  + \underbrace{\sum_{i,j} \rho d_{i,j} x_{i,j} + \sum_{i} s_i w_i}_{\mathcal{C}_2(\bx,\bw)}  \label{eq-DETobj} \\
    \text{s.t.}\quad & \sum_{i,j,r} h_{i,j}^{r} t_{i,j}^{r} \leq B,  \label{budget} \\
    & \sum_{r} t_{i,j}^{r} \leq 1,~ \forall i,j, \label{link_upgrade}\\
    & 0 \leq \sum_{i} x_{i,j} \leq C_{j}, ~ \forall j.  \label{resource_cap}\\
    &  w_i + \sum_{j} x_{i,j} = \lambda_i, ~ \forall i, \label{supply_demand} \\
    & \frac{w_{i}}{\lambda_i}\leq \alpha_i , ~ \forall i, \label{QoS}\\
    & \sum_{j} d_{i,j} \frac{x_{i,j}}{\lambda_i} \leq \Delta_{i} , ~ \forall i, \label{delay} \\
    & \bw \in \mathbb{Z}_{+}^{I}, ~\bx \in \mathbb{Z}_{+}^{I \times J}, ~\bt \in \{0,1\}^{I \times J \times R} \label{var_constr1}.
\end{align}
\end{subequations}
where $\mathcal{C}_1(\bt)$ represents the total hardening cost and $\mathcal{C}_2(\bx,\bw)$ represents the total penalty cost, including both the delay penalty and the unmet demand penalty resulting from workload allocation decisions $\bx$ and $\bw$. Note that $\rho$ is the delay penalty parameter controlled by the platform. A higher value of $\rho$ indicates that the platform prioritizes reducing delays over reducing unmet demand. %In order to ensure the feasibility of 
The underlying optimization problem of the platform is subject to constraints \eqref{budget}-\eqref{var_constr1}. 

Given that the platform's investment budget is often limited, constraint \eqref{budget} ensures that the total link hardening cost cannot exceed the budget. The platform can choose at most one hardening level for each link $(i,j)$, as shown in \eqref{link_upgrade} (see Remark~\ref{rem:one_level}). Constraints \eqref{resource_cap} specify that the total workload assigned to each EN cannot exceed its capacity. The demand from each area $i$ must be served by some EN ($\bx$) or dropped ($\bw$), as captured by \eqref{supply_demand}. 
For quality control, constraints \eqref{QoS} enforce that %the platform may impose a restriction that
the proportion of unmet demand in any area $i$ does not exceed a predefined threshold $\alpha_i$. Additionally, to maintain a satisfactory user experience, constraints \eqref{delay} impose a threshold $\Delta_{i}$ on the average network delay for each area $i$. Here, $\frac{x_{i,j}}{\lambda_i}$ represents the proportion of workload from area $i$ allocated to EN $j$. Finally, constraints \eqref{var_constr1} define all decision variables to be optimized in the model, where resource demands, allocation quantities, and unmet demands are represented as integers, as the platform typically provisions a discrete number of resource units (e.g., vCPUs, virtual machines) at each EN.

\begin{remark}
In the deterministic model, the  link delay parameter $d_{i,j}$ is %assumed to be 
a constant. Also, the hardening decision has no impact on the link delay. Hence, from \eqref{DET}, it is easy to see that the optimal solution for $\bt$ is $t_{i,j} = 0,~\forall i, j$. Then, the deterministic model \textit{DET} can be simplified as: 
\begin{align}
\label{eq-NOobj} 
     \textbf{(NH)} \quad &\min_{\bx \in \mathbb{Z}_{+}^{I \times J}, \bw \in \mathbb{Z}_{+}^{I}}  \:  \mathcal{C}_2  \\
    \text{s.t.} \quad & \eqref{resource_cap} - \eqref{delay} \nonumber.
\end{align}
There are no hardening decisions in the reduced \textit{NH} model.
\end{remark}

%\ella{Added remark below to address Reviewer III.3}
\begin{remark} \label{rem:one_level}
    The constraint \eqref{link_upgrade} restricts each link $(i,j)$ to at most a single hardening level, which simplifies the model by preclassifying and grouping hardening technologies into distinct levels. Higher levels may incorporate the protections of lower levels, thereby eliminating the need to select multiple levels simultaneously. In contrast, the alternative model of selecting multiple levels to combine techniques may require additional constraints and complicate the formulation, especially when certain combinations are infeasible or incompatible.
\end{remark}

\subsection{Uncertainty Modeling}
In practice, link delays are uncertain and can vary significantly over time. Relying on the assumption that link delays are fixed and known may lead to suboptimal solutions. 
%that can hurt user experience. 
To address this challenge, we incorporate link delays as uncertain parameters in our robust models. We propose two models for this purpose. The first model considers link delays as independent of the hardening and workload allocation decisions, resulting in an exogenous uncertainty set. In contrast, the second model utilizes a DDU set to capture the influence of decision variables on link delays. Our foremost aim is to elucidate the difference between traditional exogenous uncertainty sets and endogenous uncertainty sets, specifically regarding their impacts on the robust solution from a theoretical standpoint.

\noindent \textbf{\textit{1. Robust Hardening Model with Exogenous Uncertainties:}}
We employ the traditional RO approach to model uncertain network link delays using an exogenous polyhedral uncertainty set \cite{RObook}. In particular, the actual link delay $d_{i,j}$  is assumed to  vary within the range of $[\bar{d}_{i,j}, \bar{d}_{i,j} + \hat{d}_{i,j}]$, where $\bar{d}_{i,j}$ is the minimum  value of network delay along link $(i,j)$ and $\hat{d}_{i,j}$ is the maximum deviation from $\bar{d}_{i,j}$. Then, the link delays can be captured by the following polyhedral uncertainty set \cite{CCGARO,wang20}: 
\begin{align}\label{setD1}
    \mathcal{D}_1 (\hat{d},\bar{d},\Gamma_1) := \bigg\{ d_{i,j}:~ d_{i,j} = \bar{d}_{i,j} + g_{i,j} \hat{d}_{i,j}, \nonumber \\
     ~ g_{i,j} \in [0,1],~ \forall i,j;~ \sum_{i,j} g_{i,j} \leq \Gamma_1 \bigg\},
\end{align}
where $\Gamma_1$ is called the uncertainty budget, which can be set by the platform from observing historical data. The value of $\Gamma_1$  controls the size of the uncertainty set and the robustness of the optimal solution. As the value of  $\Gamma_1$ increases, the uncertainty set $\mathcal{D}_1$ also expands, leading to a more robust solution. The robust link hardening model  with the \textit{decision-independent uncertainty }(DIU) set $\mathcal{D}_1$ can be expressed as follows:
\begin{subequations}
\label{RO_DIU}
\begin{align}
     \textbf{(RO-DIU)} \quad &\min_{\bt,\bx,\bw} ~ \max_{\bd \in \mathcal{D}_1}  ~~  \mathcal{C}_1  + \mathcal{C}_2  \\
    \text{s.t.} \quad & \eqref{budget} - \eqref{QoS},~ \eqref{var_constr1},\\
    & \sum_{j} d_{i,j} \frac{x_{i,j}}{\lambda_i} \leq \Delta_{i}, ~ \forall i, ~\forall d_{i,j} \in \mathcal{D}_1. \label{constr_5} 
    % & t_{i,j}^{r} \in \{0,1\}, ~ t_{i,j}^{r} \leq z_j, ~\forall i,j,r \label{var_constr}
\end{align}
\end{subequations}

Compared to the deterministic model \eqref{eq-NOobj}, only minor modifications are required for the objective function and average delay constraints \eqref{constr_5} as they contain the uncertain link delay $d_{i,j}$. Similar to \textbf{DET}, the platform aims to minimize the total hardening cost while enhancing user experience. It is apparent from the  \textbf{RO-DIU} model and the exogenous uncertainty set $\mathcal{D}_1$ that link hardening provides no advantages to the platform.  As a result,  the platform has no incentives to invest in link hardening. Hence, the \textbf{RO-DIU} model can be reduced  to the  following equivalent model: 
\begin{align}
\label{eq-RONHobj} 
     \textbf{(RO-NH)} \quad &\min_{\bx \in \mathbb{Z}_{+}^{I \times J}, \bw \in \mathbb{Z}_{+}^{I}} \max_{\bd \in \mathcal{D}_1} \:  \mathcal{C}_2  \\
    \text{s.t.} \quad & \eqref{resource_cap} - \eqref{QoS},~  \eqref{var_constr1},~ \eqref{constr_5} \nonumber.
\end{align}

\noindent \textbf{\textit{2. Robust Hardening Model with Endogenous Uncertainties:}} Link hardening can enhance the quality of the edge network by reducing link delay and delay variation. However, the link delay tends to increase as data traffic (workload) traversing through the link increases. Hence, the set $\mathcal{D}_1(\hat{d},\bar{d},\Gamma_1)$ fails to capture the interdependencies between the actual link delays and link hardening as well as workload allocation decisions. To overcome this limitation, we propose a new uncertainty set that incorporates DDU to accurately model the impact of decision variables on the actual link delays. This approach enables us to better capture the effects of link hardening and workload allocation decisions on the actual link delays, thereby improving the model's robustness and efficiency. The endogenous uncertainty set representing the link delay uncertainty is defined as follows:
\begin{align}
\label{setD2}
&\mathcal{D}_2 (\hat{d},\bar{d},t,x,\Gamma_2) \! := \! \bigg\{ d_{i,j} : ~  0 \leq g_{i,j} \leq 1 - \sum_{r} \gamma_{i,j}^{r} t_{i,j}^{r}, \nonumber \\
&d_{i,j} \!=\! \bar{d}_{i,j} \!+\! \hat{d}_{i,j} ( g_{i,j} + u_{i,j} x_{i,j}), \forall i,j;~\!\sum_{i,j} g_{i,j} \leq \Gamma_2  %\label{actual_delay}\\
 %\label{harden_impact} 
    \bigg\}.
\end{align}
Here, we introduce two new decision-dependent parameters, $\gamma_{i,j}^{r}$ and $u_{i,j}$. The parameter $\gamma_{i,j}^{r}$ captures the impact of link hardening on the network delay along the link $(i,j)$ with level-$r$ hardening. As the level of hardening increases, the link delay deviation decreases. Note that at most one hardening level can be chosen for each link,  as indicated in equation \eqref{link_upgrade}.

\begin{remark}[\revtwo{On Linear Approximation of Network Delay}]
\revtwo{The network delay on a given link is influenced by its utilization, which is determined by the volume of data traffic or workload on that link. For simplicity, in the endogenous uncertainty set described in \eqref{setD2}, the network delay on each link is modeled as a linear function of its workload, parameterized by the slope coefficient $u_{i,j}$. This linear approximation is justified in scenarios where the network operates within capacity constraints and congestion is effectively managed \cite{ ChenDelay2023}.}
%, as latency components like transmission and processing delays typically scale proportionally with workload under these conditions \cite{ ChenDelay2023}. 

\revtwo{However, this assumption has inherent limitations. In scenarios with high traffic volumes or significant congestion, the relationship between latency and workload may exhibit non-linear characteristics due to various factors such as queuing delays and packet loss.  To better capture these complexities, more advanced modeling approaches or linearization techniques, such as piecewise linear approximations, can be employed. %to capture the dynamics more precisely.
The proposed model provides a foundational framework that can be extended to accommodate more complex scenarios.} 
\end{remark}

Compared to the exogenous uncertainty set \eqref{setD1} in the \textit{RO-DIU} model, the upper bound of $g_{i,j}$ in $\mathcal{D}_2$ is reduced from $1$ to $1 - \sum_{r = 1}^{R} \gamma_{i,j}^{r} t_{i,j}^{r}$, which shows the benefit of link hardening by reducing the worst-case delay on hardened links. When there is no  hardening ($\bt = \bm{0}$), the range for $g_{i,j}$ is the same for both $\mathcal{D}_1$ and $\mathcal{D}_2$. Similar to the exogenous uncertainty set \eqref{setD1}, $\Gamma_2$ is the uncertainty budget of  $\mathcal{D}_2$, which controls the size of the uncertainty set. Additionally, $\mathcal{D}_2$ captures the dependence of the link delay on the workload, which is not considered in  $\mathcal{D}_1$. Finally, it can be observed that the platform's decisions have no influence on the uncertainty set $\mathcal{D}_1$ while they can directly alter $\mathcal{D}_2$. Consequently, \textit{the platform can control the level of conservatism through its decisions}. On the other hand, the uncertainty set obviously affects the robust decisions of the platform. Hence, there is an interdependence between the decision variables and the uncertainty, which makes the robust problem with endogenous uncertainties more challenging.

Overall, the proposed robust link hardening model with the DDU set $\mathcal{D}_2$ is given as follows:
\begin{subequations}
\label{RO_DDU}
\begin{align}
 \textbf{(RO-DDU)} \quad & \min_{\bt,\bx,\bw} ~ \max_{\bd \in \mathcal{D}_2}  ~~ \mathcal{C}_1 + \mathcal{C}_2  \label{RO_DDU_obj} \\
\text{s.t.} \quad & \eqref{budget} - \eqref{QoS}, ~ \eqref{var_constr1}, \label{RO_DDU_eq1} \\
&  \sum_{j} d_{i,j} \frac{x_{i,j}}{\lambda_i} \leq \Delta_{i},  ~\forall i, ~ \forall d_{i,j} \in \mathcal{D}_2. \label{RO_DDU_delay}
\end{align}
\end{subequations}

\section{Solution Approach}
\label{sol}
\revtwo{This section develops exact reformulations to solve the robust model in \eqref{RO_DDU}, addressing the computational challenges posed by the decision-dependent nature of the uncertainties.}

\subsection{\revtwo{Overview of Solution Methods}}
The presence of decision variables, which are solutions to the outer minimization problem, in the uncertainty set $\mathcal{D}_2$ poses significant complications to solving \textit{RO-DDU}. Specifically, it changes the size and form of the set, making it difficult to reformulate robust problems with DDU sets in a tractable fashion. Indeed, the general \textit{RO-DDU} problem is NP-complete, which can be proved based on the 3-Satisfiability problem (i.e., \textbf{3-SAT}) \cite{DDU}. Please refer to \textit{Appendix~\ref{appen:npcomplete}} for the NP-Complete proof. Moreover, the uncertain parameter $d_{i,j}$ appears in both the objective function and the left-hand side of \eqref{RO_DDU_delay}, and the problem is a nonlinear program containing bilinear terms, such as $d_{i,j} x_{i,j}$, that cannot be solved directly. 

\revtwo{To address these challenges, we propose the following two solution approaches. The high-level concepts of the two algorithms are illustrated in Fig.~\ref{fig:Flowchart_ALG}.
\begin{itemize}
    \item \textit{Robust Solution (\textit{RDDU}):} This approach combines linear programming (LP) duality, binary expansion, and linearization methods to reformulate the \textit{RO-DDU} model into a large-scale MILP that can be solved by standard solvers.
    \item \textit{Enhanced Robust Solution (\textit{e-RDDU}):} This approach introduces auxiliary variables and employs a clever transformation technique to reduce the problem size by eliminating redundant constraints during reformulation, thereby substantially improving computational efficiency. 
\end{itemize}
}

\begin{figure}[t!]
\centering
	     \includegraphics[width=0.45\textwidth]{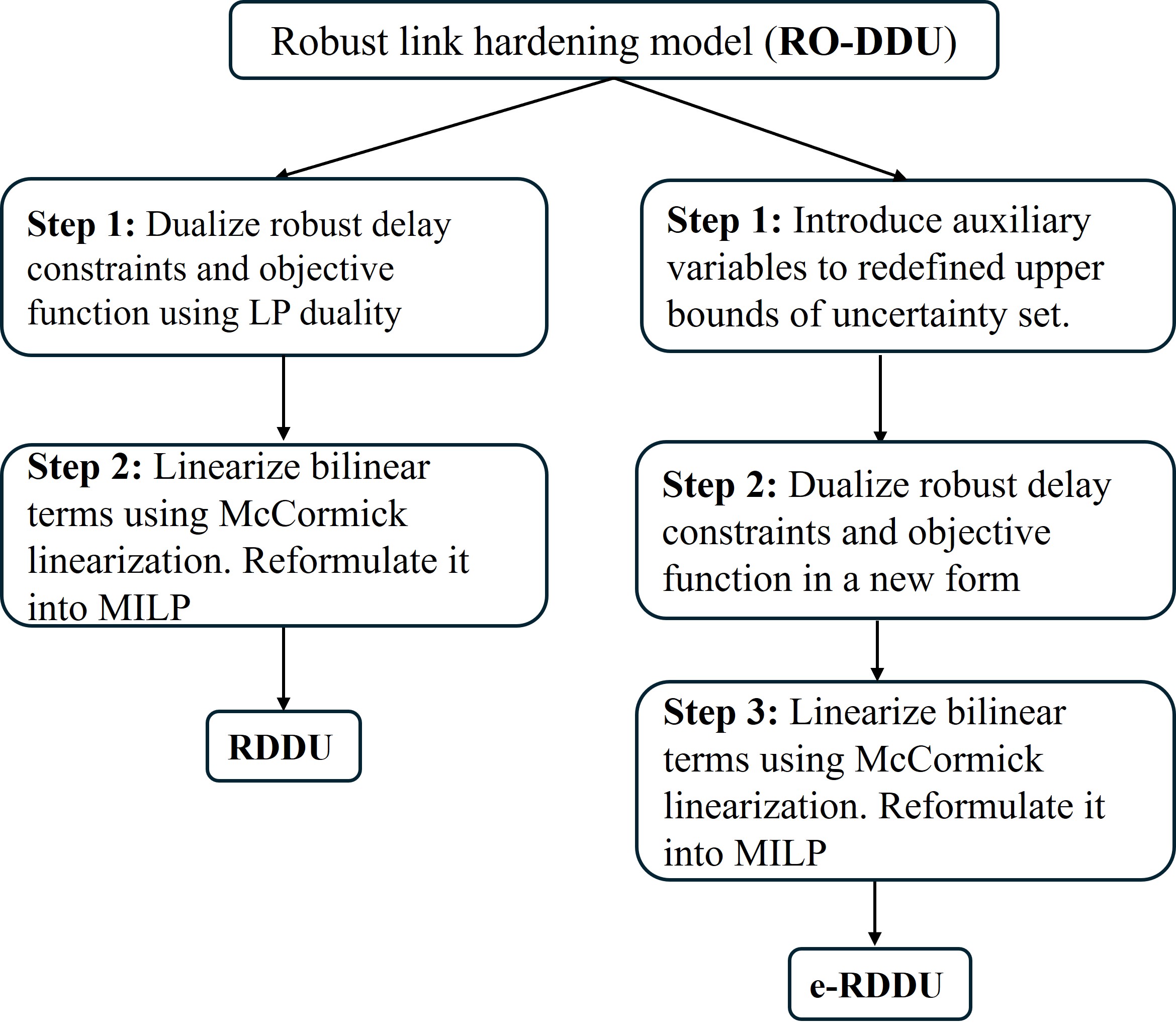}
	    \caption{\revtwo{Flowchart of solution approaches for the RO-DDU model. Two algorithms are shown: RDDU and e-RDDU.}}

     	     \label{fig:Flowchart_ALG}
             \vspace{-0.5cm}
\end{figure}

\vspace{-0.5cm}
\subsection{Robust Solution}
\label{sol:standard}
\revtwo{The main challenge in solving the  \textit{RO-DDU} problem stems from the fact that the objective function \eqref{RO_DDU_obj} and the robust constraints \eqref{RO_DDU_delay} contain bilinear terms due to the inclusion of decision variables in the DDU set $\mathcal{D}_2$. When the uncertainty set is independent of the decision variables, standard techniques can be employed to convert the problem into an equivalent MILP robust counterpart using LP duality \cite{RObook}. Unfortunately,  $\mathcal{D}_2$ is decision-dependent, rendering our problem more complicated. To this end, we introduce a series of transformations to convert % the 
\textit{RO-DDU}
into an equivalent and solvable MILP form by transforming \eqref{RO_DDU_obj} and  \eqref{RO_DDU_delay} into a set of linear equations.} 

First, the robust delay constraints  \eqref{RO_DDU_delay} can be rewritten as: 
\beqn
\Delta_{i} \geq  \max_{\bd \in \mathcal{D}_2} ~\sum_{j} d_{i,j} \frac{x_{i,j}}{\lambda_i}, ~\forall i,
\eeqn
which is equivalent to:
\beqn
\label{RO_DDU1}
    \Delta_{i}  \geq   && \sum_{j} \frac{\bar{d}_{i,j} x_{i,j}}{\lambda_i} + \sum_{j} \frac{\hat{d}_{i,j} u_{i,j}}{\lambda_i}  x_{i,j} x_{i,j} \nonumber \\
    && + ~\max_{\bm{g} \in \mathcal{G} } ~ \sum_{j} \frac{\hat{d}_{i,j}  x_{i,j}}{\lambda_i} g_{i,j}, ~ \forall i,
\eeqn
where 
\beqn
 \mathcal{G} := \Big\{g:~    \sum_{i,j} g_{i,j} \leq \Gamma_2;~0 \leq g_{i,j} \leq 1 - \sum_{r} \gamma_{i,j}^{r} t_{i,j}^{r},  \forall i,j
    \Big\}. \nonumber
\eeqn
Thus, for all $i$, the last term in \eqref{RO_DDU1} can be expressed as: 
\vspace{-0.3cm}
\begin{subequations}
\begin{align}
  \max_{\bm{g} \geq 0 } &  \qquad \sum_{j} \frac{\hat{d}_{i,j}  x_{i,j}}{\lambda_i} g_{i,j} \\
  \text{s.t.} &   \qquad  \sum_{i,j} g_{i,j} \leq \Gamma_2,   \quad \quad \quad \quad \quad \quad \quad \quad ~(\beta_i^{0})  \label{DDU_refor1} \\
   & \qquad  g_{i,j} \leq 1 - \sum_{r = 1}^{R} \gamma_{i,j}^{r} t_{i,j}^{r}, ~ \forall j, \quad \qquad ~(\xi_{i,j}^{0}) \label{DDU_refor2} 
\end{align}
\end{subequations}
 where $\beta_i^{0}$ and $\xi^{0}_{i,j}$ are dual variables associated with the constraints \eqref{DDU_refor1} and \eqref{DDU_refor2}, respectively, for each area $i$. By using LP duality \cite{LPbook}, \eqref{RO_DDU1} can be rewritten as:
\begin{subequations}
\label{RO_DDU_refor}
\begin{align}
    & \Delta_i \geq  ~\sum_{j} \bar{d}_{i,j} \frac{x_{i,j}}{\lambda_i}  + \sum_{j}\frac{ \hat{d}_{i,j} u_{i,j}}{\lambda_i} x_{i,j} x_{i,j} + \Gamma_2 \beta_i^{0} \nonumber \\
    &  \quad \quad ~ + ~\sum_{j} \xi_{i,j}^{0} -  \sum_{j,r} \gamma_{i,j}^{r} t_{i,j}^{r} \xi_{i,j}^{0}, ~ \forall i, \label{DDU_delay1}\\
    & \beta_i^{0} + \xi_{i,j}^{0} \geq  \frac{ \hat{d}_{i,j}}{\lambda_i} x_{i,j}, ~~ \forall i,j, \label{DDU_delay2} \\
    & \beta_i^{0} \geq 0, ~ \forall i; ~\xi_{i,j}^{0} \geq 0, ~ \forall i,j. \label{DDU_delay3}
\end{align}
\end{subequations}

The reformulated constraints \eqref{DDU_delay1} contain two complicating bilinear terms $x_{i,j} x_{i,j}$ and $t_{i,j}^{r} \xi_{i,j}^{0}$. If we can linearize these terms, the robust delay constraints \eqref{RO_DDU_delay} can be reformulated as a set of linear equations. We first show how to linearize the bilinear term $t_{i,j}^{r} \xi_{i,j}^{0}$, which is a product of a binary and a continuous variable. To tackle this, we define new auxiliary non-negative continuous variables $T_{i,j}^{0, r} = t_{i,j}^{r} \xi_{i,j}^{0}, \forall i,j,r$. Then, by applying McCormick linearization method \cite{McCormick76}, the bilinear term $t_{i,j}^{r} \xi_{i,j}^{0}$ can be implemented equivalently as follows:
\begin{subequations}
\label{RO_DDU_McCormick1}
\begin{align}
    & T_{i,j}^{0, r} \leq M t_{i,j}^{r},~ \forall i,j,r; ~~ T_{i,j}^{0, r} \leq \xi_{i,j}^{0}, ~~ \forall i,j,r, \\
    & 0 \leq T_{i,j}^{0, r} \geq \xi_{i,j}^{0} - M(1 - t_{i,j}^{r}), ~~ \forall i,j,r,
\end{align}
\end{subequations}
where $M$ is a sufficiently large number. It is more challenging to linearize the bilinear terms $x_{i,j}x_{i,j}$ since they are the products of two integer variables. To address this, we first introduce  binary variables $y_{i,j}^{k}$ and employ the binary expansion method to express $x_{i,j}$ as a sum of binary variables: %$y_{i,j}^{k}$'s:
\begin{subequations}
\label{xbe}
\begin{align}
    x_{i,j} & = \sum_{k=1}^{Q_{i,j}} 2^{k-1} y_{i,j}^{k}, ~ \forall i,j, \label{xbinary}\\
    & y_{i,j}^{k} \in \{0,1\}, ~ \forall i,j,k,
\end{align}
\end{subequations}
where $Q_{i,j}$ is a sufficiently large integer number. Obviously, we prefer to choose $Q_{i,j}$ as small as possible to reduce the number of binary variables $y_{i,j}^k$. Since $x_{i,j} \leq C_j, \forall i,j$ and $x_{i,j} \leq \lambda_i, \forall i,j$, we have $x_{i,j} \leq \min\{ C_j,  \lambda_i   \}, \forall i,j$. Hence, we can set $Q_{i,j} = \min \{  \floor{\log_2 C_j}+1, \floor{\log_2 \lambda_i}+1 \}$. 

We are now ready to linearize  $x_{i,j}x_{i,j}$. However, if we apply the binary expansion form \eqref{xbinary} to  $x_{i,j}$ twice, it will result in a large number of bilinear terms, each of which is a product of two binary variables. To avoid this issue, we propose a simple yet clever trick to linearize $x_{i,j}x_{i,j}$. Specifically,  instead of applying the binary expansion form  \eqref{xbe} to both $x_{i,j}$ variables, we only use it  for the first $x_{i,j}$. Therefore, we have:
\beqn
\label{xxlin}
x_{i,j}x_{i,j} = \Bigg( \sum_{k=1}^{Q_{i,j}} 2^{k-1} y_{i,j}^{k} \Bigg) x_{i,j} = \sum_{k=1}^{Q_{i,j}} 2^{k-1} Y_{i,j}^{k},~~ \forall i,j,
\eeqn
where $Y_{i,j}^{k} = y_{i,j}^{k} x_{i,j}, \forall i,j,k$. Then, we can use the McCormick linearization method again. In particular, the constraint $Y_{i,j}^{k} = y_{i,j}^{k} x_{i,j}$ can be implemented equivalently through the following linear inequalities:
\begin{subequations}
\label{RO_DDU_McCormick2}
\begin{align}
    & Y_{i,j}^{k} \leq L_{i,j} y_{i,j}^{k}; ~~  Y_{i,j}^{k} \leq x_{i,j}, ~~ \forall i,j,k, \\
    & 0 \leq Y_{i,j}^{k} \geq x_{i,j} - L_{i,j} (1 - y_{i,j}^{k}), ~~ \forall i,j,k.
\end{align}
\end{subequations}
where $L_{i,j} =  \min \{C_j, \lambda_i \}$ is an upper bound of $x_{i,j}$.

By using the preceding linearization steps, we obtain the following linear constraints that are equivalent to 
\eqref{DDU_delay1}--\eqref{DDU_delay3}.
\begin{subequations}
\label{RO_DDU_refor1}
\begin{align}
    & \Delta_i \geq  \sum_{j} \bar{d}_{i,j} \frac{x_{i,j}}{\lambda_i} + \Gamma_2 \beta_i^{0} + \sum_{j} \xi_{i,j}^{0} -  \sum_{j,r} \gamma_{i,j}^{r} T_{i,j}^{0,r}  \nonumber \\
    & + \sum_{j} \frac{\hat{d}_{i,j} u_{i,j}}{\lambda_i} \bigg( \sum_{k=1}^{Q_{i,j}} 2^{k-1} Y_{i,j}^{k} \bigg), ~ \forall i, \\
    &\eqref{RO_DDU_McCormick1},~ \eqref{xbe},~ \eqref{RO_DDU_McCormick2};~~ \beta_i^{0} + \xi_{i,j}^{0} \leq  \frac{ \hat{d}_{i,j}}{\lambda_i} x_{i,j}, ~~ \forall i,j, \\
    & \beta_i^{0} \geq 0, ~ \forall i; ~\xi_{i,j}^{0} \geq 0, ~ \forall i,j.
\end{align}
\end{subequations}
Since the robust delay constraints \eqref{RO_DDU_delay} are equivalent to \eqref{RO_DDU_refor}, they are also equivalent to \eqref{RO_DDU_refor1}.

To linearize the robust objective function in 
\eqref{RO_DDU_obj}, we first rewrite $\min_{\bt,\bx,\bw} \max_{\bd \in \mathcal{D}_2}  (\mathcal{C}_1 + \mathcal{C}_2 )$ as:
\begin{subequations}
\label{RO_DDU_eta}
\begin{align}
\min_{\bt,\bx,\bw} & \qquad \eta  \\
\text{s.t.} & \qquad \eta \geq ~ \max_{\bd \in \mathcal{D}_2}  ~( \mathcal{C}_1 + \mathcal{C}_2 ). \label{dduobj_eta}
\end{align}
\end{subequations}

From \eqref{setD2}, we can express \eqref{dduobj_eta} as:
\beqn
\label{RO_DDU_obj1}
 \eta~ \geq  && \mathcal{C}_1 + \sum_{i} s_{i} w_i + \rho \sum_{i,j} \bar{d}_{i,j} x_{i,j}  \\   \nonumber 
     &&  +~ \rho \sum_{i,j} \hat{d}_{i,j} u_{i,j} x_{i,j} x_{i,j} + \max_{\bm{g} \in \mathcal{G}} ~\rho \sum_{i,j} \hat{d}_{i,j} g_{i,j} x_{i,j}.
\eeqn

The last term in \eqref{RO_DDU_obj1} can be written explicitly as:
\begin{subequations}
\label{RO_DDU_obj2}
\begin{align}
     \max_{\bm{g} \geq 0} \quad & \rho \sum_{i,j} \hat{d}_{i,j} g_{i,j} x_{i,j} \\
 \text{s.t.} \quad & \sum_{i,j} g_{i,j} \leq \Gamma_2, \qquad\qquad\qquad \qquad ~~(\beta^{1})  \label{obj_DDU1}  \\
    & ~ g_{i,j} \leq 1 - \sum_{r = 1}^{R} \gamma_{i,j}^{r} t_{i,j}^{r},~\forall i,j, ~\qquad (\xi_{i,j}^{1}) \label{obj_DDU2}
\end{align}
\end{subequations}
where $\beta^{1}$ and $\xi_{i,j}^{1}$ are the dual variables associated with %constraints 
\eqref{obj_DDU1} to \eqref{obj_DDU2}, respectively. Define $T_{i,j}^{1,r} = t_{i,j}^r \xi_{i,j}^1, \forall i,j,r$. By applying the strong duality theorem \cite{LPbook} to \eqref{RO_DDU_obj2}, and using \eqref{xbe}-\eqref{RO_DDU_McCormick2}, and the McCormick linearization method,  we have:
\begin{subequations}
\label{RO_DDU_reform}
\begin{align}
    % & \qquad\qquad\qquad\min_{\bt,\bx,\bw,\eta} \qquad \eta  \\
    % &  \text{s.t.}~~~ 
    % &\eqref{RO_DDU_McCormick1},~ \eqref{RO_DDU_McCormick2} \\
    & \eta \geq  ~ \mathcal{C}_1 + \sum_{i} s_{i} w_i  + \rho \sum_{i,j} \hat{d}_{i,j} u_{i,j} \bigg( \sum_{k=1}^{Q_{i,j}} 2^{k-1} Y_{i,j}^{k} \bigg) \nonumber \\
    & + \rho \sum_{i,j} \bar{d}_{i,j} x_{i,j} + \Gamma_2 \beta^1 + \sum_{i,j} \xi_{i,j}^{1} - \sum_{i,j,r} \gamma_{i,j}^{r} T_{i,j}^{1,r},\\
    & \eqref{xbe}, \eqref{RO_DDU_McCormick2};~ \beta^{1} + \xi_{i,j}^{1} \geq \rho \hat{d}_{i,j} x_{i,j}, ~ \forall i,j, \\
    &  T_{i,j}^{1,r} \leq M  t_{i,j}^{r}, ~~ \forall i,j,r; ~~  T_{i,j}^{1,r} \leq \xi_{i,j}^{1}, ~~ \forall i,j,r, \\
    & 0 \leq T_{i,j}^{1,r} \geq \xi_{i,j}^{1} - M (1 - t_{i,j}^{r}), ~~ \forall i,j,r, \\
    & \beta^{1} \geq 0; ~\xi_{i,j}^{1} \geq 0, ~ \forall i,j.
\end{align}
\end{subequations}
Finally, the \textit{RO-DDU} problem \eqref{RO_DDU} can be reformulated as the following MILP, which can be solved by off-the-shelf solvers.
\begin{subequations}
\label{RO_DDU1MILP}
\begin{align}
\textbf{(RDDU)} \quad  &\min_{\bt,\bx,\bw}  \quad \eta  \\
\text{s.t.} \quad &  \eqref{budget} - \eqref{QoS}, ~ \eqref{var_constr1},~ \eqref{RO_DDU_refor1},~ \eqref{RO_DDU_reform}.
\end{align}
\end{subequations}

\vspace{-0.5cm}
\subsection{Enhanced Robust Solution}
\revtwo{Although the MILP reformulation \eqref{RO_DDU1MILP} can solve the decision-dependent robust problem (\textit{RO-DDU}) in \eqref{RO_DDU}, it suffers from weak relaxations due to the big-M terms in  \eqref{RO_DDU_McCormick1}, leading to poor numerical performance. Furthermore, the size of the reformulation \textit{RDDU}  increases rapidly as the problem size, especially the upper bound of $x_{i,j}$, increases. To this end, we propose an alternative reformulation that can substantially improve the computational time.} Specifically, we introduce new auxiliary variables and employ an alternative but equivalent form of $\mathcal{D}_2$ to reformulate the original \textit{RO-DDU} problem. %\textbf{The key difference} between the enhanced robust solution and the robust solution in Section \ref{sol:standard} lies at the following step where
We first  rewrite the constraints on $g_{i,j}$ in $\mathcal{D}_2$  as follows:
\begin{align}
\label{trans}
     0 \leq g_{i,j} \leq \bigg(1 - \sum_{r} \gamma_{i,j}^{r}\bigg) + \sum_{r} \gamma_{i,j}^{r} (1 - t_{i,j}^{r}), ~ \forall i,j.
\end{align}
Define new binary variables $v_{i,j}^{r} = 1 - t_{i,j}^{r}, \forall i,j,r$. Then:
\beqn
\label{setD2a}
  &&   \mathcal{D}_2 (\hat{d},\bar{d},t,x,\Gamma_2)  
    := \big\{ d : ~ g \in \mathcal{G}', \nonumber \\ &&\!d_{i,j} \!=\! \bar{d}_{i,j} \!+\! \hat{d}_{i,j} \big( g_{i,j} + u_{i,j} x_{i,j}\big), \forall i,j 
    \big\},
\eeqn
\vspace{-0.5cm}
\beqn
\label{setG2}
\!\!\!\!\!\!\!\!\!\!\!\!\!\!\!\!\text{where}&&\!\!\!\!\!\!\!\!\mathcal{G}' := \bigg\{~ \sum_{i,j} g_{i,j} \leq \Gamma_2,  \nonumber\\   
 &&\!\!\!\!\!\!\!\! 0 \leq g_{i,j} \leq \bigg(1 - \sum_{r} \gamma_{i,j}^{r}\bigg) + \sum_{r} \gamma_{i,j}^{r} v_{i,j}^r, ~ \forall i,j  \bigg\}. 
\eeqn

It can be observed that \textbf{the key difference} between the enhanced robust solution and the robust solution in Section \ref{sol:standard} lies in the introduction of new binary variables ($v_{i,j}^{r} = 1 - t_{i,j}^{r}$) and replacing the upper bound of $g$ in \eqref{setD2} with \eqref{trans}.
Now, the robust delay constraint \eqref{RO_DDU_delay} can be expressed as:
\beqn
\label{M_delay}
    \Delta_{i} \geq  \sum_j \bar{d}_{i,j} \frac{x_{i,j}}{\lambda_i} + \sum_{j} \hat{d}_{i,j}  \frac{u_{i,j}}{\lambda_i} x_{i,j} x_{i,j}  \nonumber \\
    + \max_{\bm{g} \in \mathcal{G}'} ~ \sum_{j} \hat{d}_{i,j} \frac{x_{i,j}}{\lambda_i} g_{i,j}, \forall i.
\eeqn
The last term in \eqref{M_delay} can be given explicitly as:
\begin{subequations}
\label{new_set_constr}
\begin{align}
  \max_{\bm{g} \geq 0 } &  \qquad \sum_{j} \frac{\hat{d}_{i,j}  x_{i,j}}{\lambda_i} g_{i,j} \\
  \text{s.t.} &   \qquad  \sum_{i,j} g_{i,j} \leq \Gamma_2,   \quad \quad \quad \quad \quad \quad \quad \quad \quad ~~~(\beta_i^{2})  \label{DDU_refor12} \\
   & \qquad    g_{i,j} \leq (1 - \sum_{r} \gamma_{i,j}^{r}) + \sum_{r} \gamma_{i,j}^{r} v_{i,j}^r,\forall j, (\xi_{i,j}^{2})   \label{DDU_refor22}  
\end{align}
\end{subequations}
 where $\beta_i^{2}$ and $\xi_{i,j}^{2}$ are dual variables associated with the constraints \eqref{DDU_refor12} and \eqref{DDU_refor22}, respectively, for each area $i$. 
Using the strong duality for the preceding problem yields:
\begin{subequations}
\label{M_delay_dual}
\begin{align}
    & \!\!\Delta_{i} \geq \sum_{j} \bar{d}_{i,j} \frac{x_{i,j}}{\lambda_i} + \Gamma_2 \beta_i^{2} + \sum_{j} \xi_{i,j}^{2} + \sum_{j,r} \gamma_{i,j}^{r}  v_{i,j}^{r} \xi_{i,j}^{2}  \nonumber \\
    &~~  - \sum_{j,r} \gamma_{i,j}^{r} \xi_{i,j}^{2} \!+\! \sum_{j} \hat{d}_{i,j} \frac{u_{i,j}}{\lambda_{i}} \bigg( \!\sum_{k=1}^{Q_{i,j}} \! 2^{k-1} y_{i,j}^{k} \bigg) x_{i,j}, \forall i,\!\!\!\!\\
    & \beta_i^{2} +\xi_{i,j}^{2} \geq \frac{\hat{d}_{i,j}}{\lambda_{i}} x_{i,j}, ~ \forall i,j, \label{M_delay1}\\
    & \beta_i^{2} \geq 0, ~ \forall i; ~\xi_{i,j}^{2} \geq 0, ~ \forall i,j; ~ y_{i,j}^{k} \in \{0,1\}, ~ \forall i,j,k. \label{M_delay2}
\end{align}
\end{subequations}
The constraints in \eqref{M_delay_dual} can be rewritten by expanding the variable space as:
\begin{subequations}
\label{MM_delay_dual}
\begin{align}
     & \Delta_{i} \geq \sum_{j} \bar{d}_{i,j} \frac{x_{i,j}}{\lambda_i} + \Gamma_2 \beta_i^{2} + \sum_{j} \xi_{i,j}^{2} + \sum_{j,r} V_{i,j,r}^{2} \nonumber \\
     &~~ - \sum_{j,r} \gamma_{i,j}^{r} \xi_{i,j}^{2} + \sum_{j} \frac{\hat{d}_{i,j} u_{i,j}}{\lambda_{i}} \sum_{k=1}^{Q_{i,j}} 2^{k-1} Y_{i,j}^{k} , ~ \forall i,\!\!\!\!\\
     & V_{i,j,r}^{2} \geq \gamma_{i,j}^{r} v_{i,j}^{r} \xi_{i,j}^{2}, \forall i,j,r; ~~ Y_{i,j}^{k} \geq y_{i,j}^{k} x_{i,j}, \forall i,j,k, \label{MM_delay3}\\
    & \eqref{M_delay1},~ \eqref{M_delay2},
\end{align}
\end{subequations}
where the auxiliary non-negative continuous variables $V_{i,j,r}^{2}$ and $Y_{i,j}^{k}$ related to the original variables through the constraints \eqref{MM_delay3}, respectively. If dual variable $\bm{\xi}$ is feasible for the set of equations given by \eqref{M_delay_dual}, then we can find a feasible variable for \eqref{MM_delay_dual} by $V_{i,j,r}^{2} = \gamma_{i,j}^{r} v_{i,j}^{r} \xi_{i,j}^{2}$, $Y_{i,j}^{k} = y_{i,j}^{k} x_{i,j}$ and vice versa. From constraints \eqref{MM_delay3}, if $v_{i,j}^{r} = 0$, then $V_{i,j,r}^{2} \geq 0$ and if $v_{i,j} = 1$, then $V_{i,j,r}^{2} \geq \gamma_{i,j}^{r} \xi_{i,j}^{2}$. Thus, $V_{i,j,r}^{2} \geq \gamma_{i,j}^{r} v_{i,j}^{r} \xi_{i,j}^{2}$ can be expressed by the following set of constraints:
\begin{align}
\label{RO_DDU_McCormick3}
    0 \leq V_{i,j,r}^{2} \geq ~ \gamma_{i,j}^{r} \xi_{i,j}^{2} - M (1 - v_{i,j}^{r}), ~~ \forall i,j,r.
\end{align}
The variables $Y_{i,j}^{k}$ satisfy the linearized constraints in \eqref{RO_DDU_McCormick2} as presented in the previous section.

Similarly, by following same procedure for objective function $\eta = \max_{\bd \in \mathcal{D}_2}  \: \mathcal{C}_1 + \mathcal{C}_2$, the robust counterpart of $\eta$, the linearized conditions can be obtained, as follows
\begin{subequations}
\label{MM_obj_dual}
\begin{align}
    & \!\!\eta \geq  \mathcal{C}_1 + \sum_{i} s_{i} w_i + \rho \sum_{i,j} \bar{d}_{i,j} x_{i,j} + \Gamma_2 \beta^3 + \sum_{i,j} \xi_{i,j}^{3} \nonumber \\
    &  -\! \sum_{i,j,r} \gamma_{i,j}^{r} \xi_{i,j}^{3} +\!\! \sum_{i,j,r} V_{i,j,r}^{3}  + \rho \sum_{i,j} \hat{d}_{i,j} u_{i,j} \!\sum_{k=1}^{Q_{i,j}} 2^{k-1} Y_{i,j}^{k}, \!\!\!\!\\
    & \eqref{xbe},~\eqref{RO_DDU_McCormick2}; ~ \beta^{3} + \xi_{i,j}^{3} \geq \rho \hat{d}_{i,j} x_{i,j}, ~ \forall i,j, \\
    & \beta^{3} \geq 0; ~\xi_{i,j}^{3} \geq 0, ~ \forall i,j,\\
    & 0 \leq V_{i,j,r}^{3} \geq \gamma_{i,j}^{r} \xi_{i,j}^{3} - M (1 - v_{i,j}^{r}), ~ \forall i,j,r.
\end{align}
\end{subequations}

Overall, the \textit{RO-DDU} problem in \eqref{RO_DDU} can be reformulated as the following equivalent MILP:
\begin{subequations}
\label{RO_DDU2}
\begin{align}
\textbf{(e-RDDU)} ~  &\min_{\bt,\bx,\bw}  \qquad \eta  \\
\text{s.t.} \quad &  \eqref{budget} - \eqref{QoS}, ~ \eqref{var_constr1}, ~ \eqref{MM_delay_dual},~ \eqref{MM_obj_dual}.
\end{align}
\end{subequations}

\subsection{Equivalence of RDDU and e-RDDU}
\revtwo{In the following, we offer the intuition for the proposed reformulations  \textit{RDDU} and \textit{e-RDDU}, and demonstrate their equivalence. This will enable us to understand the computational improvements provided by the \textit{e-RDDU} reformulation in comparison to the \textit{RDDU} reformulation.} Consider the following polyhedral DDU set with affine decision dependence: 
\begin{align}
\label{dduset}
    \mathcal{U}(\bz) := \big\{ \bm{\zeta} | \bm{A} \bm{\zeta} \leq \bv + \bm{\psi} \bz  \big\},
\end{align}
where $\bm{A}$ is a constant matrix, $\bv$ is a constant vector, and $\psi$ is an impact matrix with appropriate dimensions determining the influence of the decision $\bz$ on the upper bound of the uncertainty. Here, $\bz$ is a vector of binary variables (0's and 1's). Note that since $\mathcal{D}_2$ contains the integer variables $x_{i,j}$'s, we employed the binary expansion method as shown in \eqref{xbe} so that $\mathcal{D}_2$ encompasses only binary and continuous variables.  Thus, $\mathcal{D}_2$  is a special case of the DDU set $\mathcal{U}(\bz)$ in \eqref{dduset}.
Consider a general robust linear constraint given by:
\begin{align} \label{eq:MaxRobustConstOrg}
\bm{\zeta}^{T}\bu  \leq \bb, ~ \forall \bm{\zeta} \in \mathcal{U}(\bz).
\end{align}

\begin{theorem}
\label{Theo:RDDU}
Constraint \eqref{eq:MaxRobustConstOrg} can be reformulated as:
\begin{subequations}\label{eq:TheoRDDU}
\begin{align}
    \bm{\pi}^{T} \bm{A} = \bu^{T}, ~~ \bm{\pi} \geq 0,\label{eq:TheoRDDU_2}\\ 
    \sum_i\pi_i v_i + \sum_i\sum_j\psi_{i,j}y_{i,j} \leq b, \label{eq:TheoRDDU_1} \\
    y_{i,j} \leq \pi_i, ~ y_{i,j} \leq M z_j, ~\forall i,j \label{eq:TheoRDDU_3},\\
    0 \leq y_{i,j} \geq \pi_i - M (1 - z_j), ~\forall i,j, \label{eq:TheoRDDU_4}
\end{align}
\end{subequations}
where $M$ is a sufficiently large constant.
\end{theorem}

\begin{proof} Please refer to \textit{Appendix} B.  %\ref{proofRDDU}.
\end{proof}
\textit{Theorem} \ref{Theo:RDDU} provides the robust counterpart \eqref{eq:TheoRDDU} for the robust constraint \eqref{eq:MaxRobustConstOrg} using the Big-M reformulation. While this method has the advantage of not requiring a special set structure, it may lead to weak relaxations
and a rapidly increasing reformulation size as the number of decision
variables $z_j$'s increases.  This is a significant issue in our problem since $\mathcal{D}_2$ contains a large number of binary variables. To address this issue, we propose an enhanced reformulation by reducing the number of Big-M constraints by manipulating the coefficient matrix of binary variables, such that all elements are non-negative. The following theorem presents
the result.
\begin{theorem}
\label{Theo:e-RDDU}
The robust constraint \eqref{eq:MaxRobustConstOrg} is equivalent to: 
\begin{subequations}\label{eq:Theoe-RDDU}
\begin{align}
    \bm{\pi}^{T} \bm{A} = \bu^{T}, ~~ \bm{\pi} \geq 0,\label{eq:Theoe-RDDU_2}\\ 
    \sum_i\pi_i (v_i +\sum_{j:\psi_{i,j} < 0}\psi_{i,j}) + \sum_i\sum_{j:\psi_{i,j} \ge 0}\psi_{i,j}y_{i,j}~~~~~\cr
    - \sum_i\sum_{j:\psi_{i,j} < 0}\psi_{i,j}w_{i,j} \leq b, \label{eq:Theoe-RDDU_1} \\
    0 \leq y_{i,j} \geq \pi_i - M (1 - z_j),~\forall i,j : \psi_{i,j} \ge 0, \label{eq:Theoe-RDDU_4}\\
    0 \leq w_{i,j} \geq \pi_i - M  z_j,~\forall i,j : \psi_{i,j} < 0. \label{eq:Theoe-RDDU_5}
\end{align}
\end{subequations}
\end{theorem}
\begin{proof}
Please refer to \textit{Appendix} C. %\ref{proofe-RDDU}.
\end{proof}

It is easy to see that the \textit{RDDU} reformulation relies on the result of \textit{Theorem} \ref{Theo:RDDU}, while \textit{e-RDDU} exploits  \textit{Theorem} \ref{Theo:e-RDDU} to reduce the number of Big-M constraints. Therefore, we can infer that \textit{RDDU} and \textit{e-RDDU} are equivalent. Furthermore, the computational advantages of \textit{e-RDDU} over \textit{RDDU} stem from the fact that \eqref{eq:Theoe-RDDU} has considerably fewer integer constraints compared to  \eqref{eq:TheoRDDU}.

\subsection{Comparison Between RDDU and e-RDDU} 

\revtwo{To highlight the advantages of \textit{e-RDDU},  we compare the problem sizes resulting from the two reformulations.} Recall from \eqref{xbe} that $k$ represents the index in the set $[1,Q_{i,j}], \forall i,j$. 
Define $K = \sum_{i,j} Q_{i,j}$. Also, set $\mathcal{E}$ represents all the network links between APs and ENs, i.e., $\mathcal{E}=\{(i,j):i\in \mathcal{I},j\in \mathcal{J}\}$. We use $|.|$ to indicate the cardinality of a set. For instance, $|\mathcal{E}|$ is the total number of links in the network. Since both the reformulations \textit{RDDU} and \textit{e-RDDU} have similar numbers of continuous and integer variables, the difference between \textit{RDDU} and \textit{e-RDDU} lies in the number of constraints, which becomes apparent as network size increases. \revtwo{Table \ref{tab: Var_table} shows the sizes of the \textit{RDDU} and \textit{e-RDDU} problems in terms of the number constraints, which demonstrates the advantages of \textit{e-RDDU}.}

\renewcommand{\arraystretch}{1.5}
\begin{table}[t!]
\centering
\begin{tabular}{ |c|c|}
\hline
 Method & Number of constraints\\
 \hline
 \!\!\!\! \textit{RDDU} \!\!\!   & \!\!\!\! $ 2 + 2|\mathcal{I}| \!+\! |\mathcal{E}| (8|\mathcal{R}| \!+\! 5K \!+ \!\!5)| \!+\! |\mathcal{F}|\!\!\!$  \!\!  \\
 \hline 
 \!\!\!\!\! \textit{e-RDDU} \!\!\! & \!\!\!\!\!\!  $ 2 + 2|\mathcal{I}| +\! |\mathcal{E}| (4|\mathcal{R}| \!+\! 5K \!+\!\! 5) \!+\! |\mathcal{F}|$ \!\!\!\!  \\
 \hline 
\end{tabular}
\caption{\revtwo{Comparison between \textit{RDDU} and \textit{e-RDDU}}}
\label{tab: Var_table}
\vspace{-0.4cm}
\end{table}
\renewcommand{\arraystretch}{1}

\begin{remark}[\revtwo{Computational Complexity}]
The proposed \textit{RO-DDU} model is a complex mixed-integer nonlinear program, proven to be NP-complete in Appendix~\ref{appen:npcomplete}, making it challenging to solve in its original form. The \textit{RDDU} and \textit{e-RDDU} reformulations transform the original problem into an MILP, solvable by off-the-shelf solvers. Consequently, the overall complexity of solving \textit{RO-DDU} remains dependent on the solver's choice and its computational capability. MILP itself is a challenging class of optimization problems, being NP-hard, and when treated as a decision problem, it is NP-complete \cite{conforti2014integer}. Most MILP solvers incorporate heuristic algorithms during a pre-solve phase to effectively reduce the problem size before employing techniques such as branch-and-bound and branch-and-cut \cite{conforti2014integer}. However, the computational time for these methods is primarily influenced by the problem size, particularly the number of integer variables, and the worst-case complexity remains exponential.
\end{remark}

\section{Numerical Results}
\label{results}
This section evaluates the performance of our proposed models through simulations, including comparisons with benchmarks and sensitivity analyses to examine the impact of key system parameters on the solution  and system performance.

\subsection{Simulation Setting}
\label{simsetting}
We consider an EC  system comprising $10$ areas and $10$ ENs (i.e., $I$ = $J$ = $10$), with larger networks explored in our sensitivity analyses. The edge network topology is based on the cities and locations of randomly selected Equinix edge data centers\footnote{https://www.equinix.com/data-centers/americas-colocation}, as similarly done in related work \cite{Hotedge}. The network delay parameters, including the minimum delay $\bar{d}_{i,j}$ and maximum delay deviation $\hat{d}_{i,j}$, are generated using the global ping dataset\footnote{ https://wondernetwork.com/pings}. Resource demands $\lambda_i$'s are generated using the Materna data trace\footnote{http://gwa.ewi.tudelft.nl/datasets/gwa-t-13-materna}, which includes performance metrics from distributed data centers, such as CPU cores, memory usage, and disk throughput over three months. This facilitates the simulation of realistic resource demands, with values ranging between $40$ and $60$ \textit{vCPUs}, randomly generated based on a uniform distribution. The resource capacities of the ENs are randomly selected based on specifications of EC2 instances\footnote{https://instances.vantage.sh.}. Each EN consists of several EC2 instances. The unmet demand penalty parameters $s_i$ are generated using a uniform distribution over the interval $[40,50]$, which captures the variability in penalties that operators might incur depending on the criticality of service requirements and user expectations.

We consider three hardening levels for each link (i.e., $R = 3$). % in the simulation.
The link hardening cost parameters (i.e., $h_{i,j}^{r}, \forall i,j,r$) are generated as follows.
Recall that $\Delta h$ is denoted by the incremental link hardening cost between two adjacent levels for each link, which is set to $0.2$ in our simulation. The first-level hardening costs $h_{i,j}^{\sf 1}$ are randomly generated from a uniform distribution within the interval $[1, 1.05]$. The second and third-level link hardening costs are obtained by adding the incremental cost to the first-level hardening costs, i.e.,  $h_{i,j}^{\sf 2} = h_{i,j}^{\sf 1} + \Delta h$ and $h_{i,j}^{\sf 2} = h_{i,j}^{\sf 1} + 2 \Delta h, ~\forall i,j$. These incremental costs are set to reflect realistic trade-offs between cost and performance improvement, consistent with established network hardening practices \cite{US_harden1,US_harden2}.
The parameters $\gamma_{i,j}^r$, which represent the impact of the hardening decision on the link delay, are generated using a similar method. Specifically, we set $\gamma_{i,j}^{1}, \forall i,j$ to  $0.1$ and define $\Delta \gamma$ as an incremental impact factor (IF) between adjacent levels, which is set to $0.4$. Then, we have $\gamma^2 = \gamma^1 + \Delta \gamma$ and $\gamma^3 = \gamma^1 + 2 \Delta \gamma$. 
Thus, the IFs for all levels of hardening can be determined using only the IF for first-level hardening decisions ($\gamma_{i,j}^{r =1 }$) and the incremental IF ($\Delta \gamma$), which simplifies conducting sensitivity analysis. It is important to note that the IF $\gamma_{i,j}^r$ falls within the range $[0,1]$, meaning the highest level, $\gamma_{i,j}^{r = 3}$, must be less than $1$. To simplify, we assign distinct IFs within this range for the three hardening levels. The impact of the allocated workload is assumed to be uniform across all links, with the base case setting $u_{i,j} = \bu = 0.1, \forall i, j$. 

\begin{table}[t!]
\centering
\begin{tabular}{|l|l|l|l|}
\hline
Parameter               & Value & Parameter                 & Value           \\ \hline
$I, J$               & $10$    & $B$                    & 100             \\ \hline
$R$                  & $3$     & $\Gamma_1 = \Gamma_2 $ & 15              \\ \hline
$\gamma_{i,j}^{r=1}$ & $0.1$   & $h_{i,j}$              & $U[1, 1.05]$ \\ \hline
$\Delta \gamma$      & $0.4$   & $\Delta h$             & 0.2             \\ \hline
$\rho$               & $0.1$   & $\Delta_i = \Delta$    & 15              \\ \hline
$\alpha$             & $0.05$  & $u_{i,j}$              & $0.1$             \\ \hline
\end{tabular}
\caption{System parameters}
\label{tab:SystemParameters}
\vspace{-0.4cm}
\end{table}

The system parameter values are detailed in Table~\ref{tab:SystemParameters}, which represents the \textbf{default setting}.During the sensitivity analysis, these key parameters will be varied to evaluate their impact on system performance. All the experiments are implemented in a Matlab environment using CVX\footnote{https://cvxr.com/cvx/} and Gurobi 9.1.2 on a desktop with an Intel Core i7-11700KF and 32 GB of RAM.  In CVX, we use the default precision setting, where the relative accuracy for the optimality conditions is approximately $10^{-6}$ to $10^{-8}$. For Gurobi solver, the default tolerance, i.e., \textit{MIPGap} for MILP, is set to around $10^{-4}$ for mixed integer programs.

\subsection{Sensitivity analysis}
\label{sensi}

This section presents sensitivity analyses to assess the influence of key system parameters on the optimal solution. These parameters  include the hardening budget ($B$), uncertainty budget ($\Gamma$), delay penalty ($\rho$), and DDU parameters ($\bf{u}$, $\bf{\gamma}$). To evaluate the impact of the link hardening cost $\bf{h}$, a scaling factor $\Psi$ is introduced for the cost, where $\Psi$ is equal to  $1$ in the default setting. More specifically, the base value of $\bf{h}$ generated in Section \ref{simsetting} is multiplied by  $\Psi$ to either scale up or down the hardening cost. A higher value of $\Psi$ indicates a higher cost for hardening each link. 
For the purpose of sensitivity analyses, we solely focus on the proposed robust model  \textit{RDDU}. The optimal objective value expresses the total cost. 

\subsubsection{Impacts of the link hardening costs and budget}
As illustrated in Figs.\ref{fig:B_Psi_cost} and \ref{fig:B_Psi_payment}, increasing the hardening budget B leads to a decrease in the total cost since the platform can select more links for hardening, resulting in a reduction of the total cost. It is important to note that the delay constraints \eqref{delay} prioritize the closest ENs to meet the demand of each area, leading to only a subset of logical links being utilized. Clearly, the platform should only harden links with traffic traversing through them. Thus, as shown in Fig.~\ref{fig:B_Psi_cost}, the cost becomes saturated after a certain budget value since the platform has already chosen the most critical links for hardening and has no incentive to harden more links even if the budget is redundant. 

Fig. \ref{fig:B_Psi_cost} further shows that the total cost increases as the link hardening cost scaling factor $\Psi$ increases.  The saturation occurs early when $\Psi$ is small. For lower values of $\Psi$, the platform is more inclined to invest in link hardening, leading to an initial increase in hardening payment followed by saturation, as confirmed in Fig.~\ref{fig:B_Psi_payment} where the total link hardening cost (i.e., payment) becomes saturated after a certain budget value. 

\begin{figure}[h!]
\vspace{-0.2cm}
\centering
		\subfigure[Total cost]{
	     \includegraphics[width=0.242\textwidth,height=0.10\textheight]{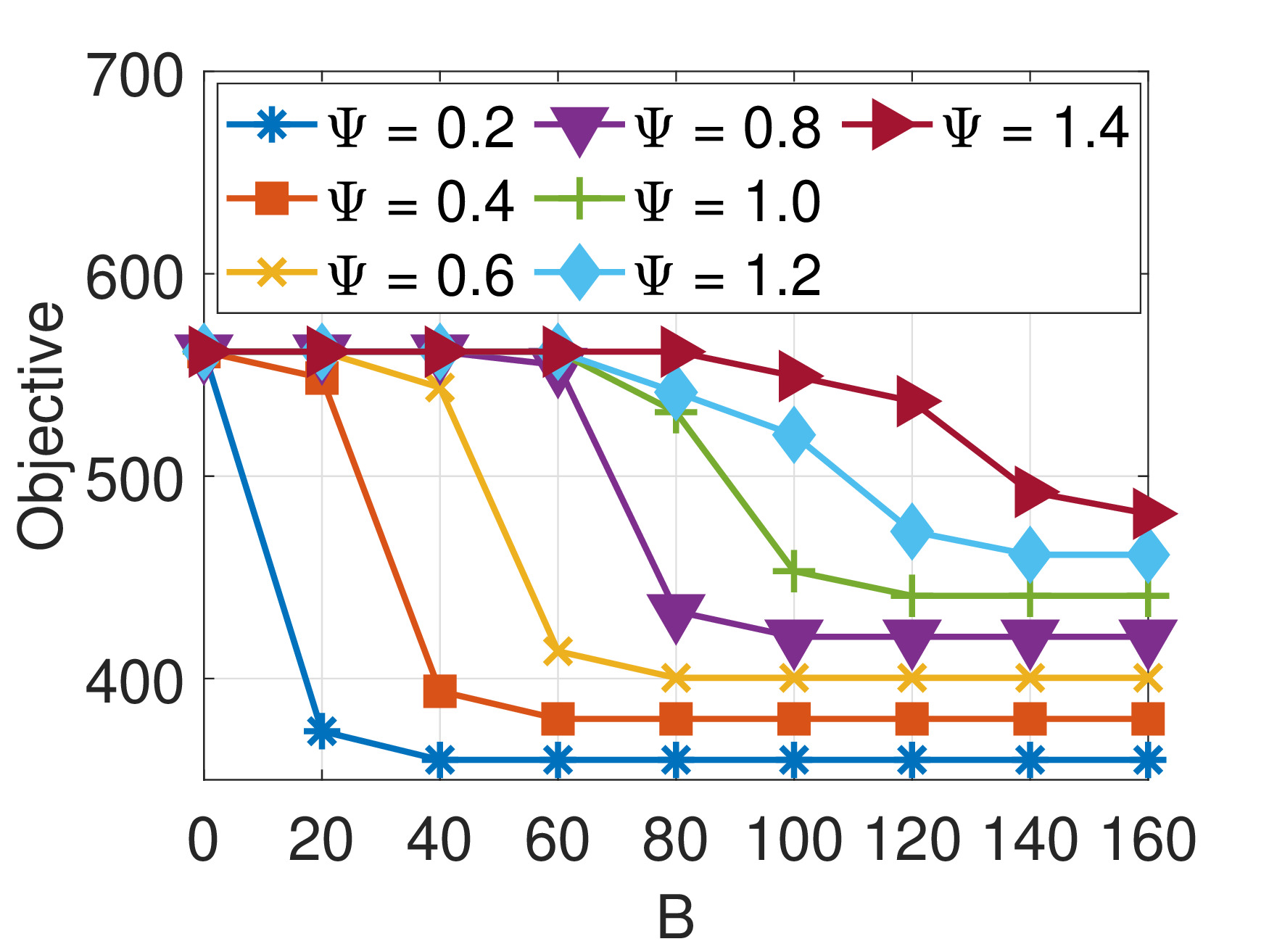}
	     \label{fig:B_Psi_cost}
	}  \hspace*{-2.1em} 
	     \subfigure[Payment]{
	     \includegraphics[width=0.242\textwidth,height=0.10\textheight]{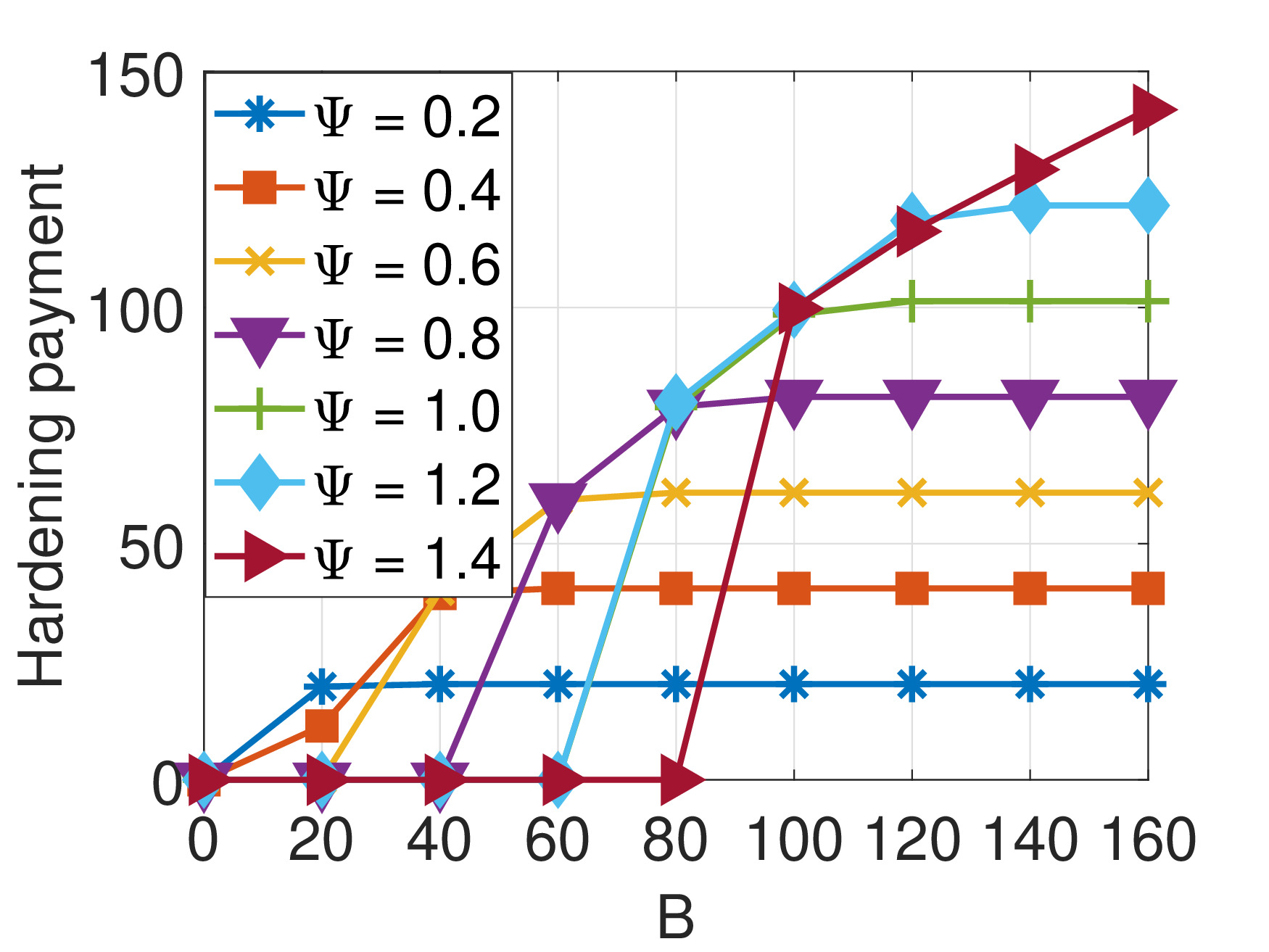}
	     \label{fig:B_Psi_payment}
	}
	\vspace{-0.2cm}
	    \caption{Impacts of $\Psi$ and $B$ on the system performance}
	    \vspace{-0.4cm}
\end{figure}

\subsubsection{Impacts of the uncertainty set}
\label{impact_uncertainty}
Figs. \ref{fig:Gamma_rho}-\ref{fig:u_Gamma} exhibit the effects of the uncertainty set on the optimal solution. It is evident from Fig.~\ref{fig:Gamma_rho} that the total cost initially increases with the increment of $\Gamma$, resulting in a larger uncertainty set and a more robust solution. However,  the total cost reaches saturation after a certain value of $\Gamma$, as the platform has already chosen an optimal set of links for hardening. Note that only a subset of ENs can serve demand from a specific area, causing certain logical links to remain unused for data traffic delivery. Therefore, the platform has no incentive to harden these links. Recall that $\bf{u}$ represents the influence of workload on link delay, as specified in $\mathcal{D}_2$. Fig. \ref{fig:u_Gamma} suggests that a higher value of $\bf{u}$ corresponds to a greater link delay variation, leading to a higher total cost. These results are intuitive, as a larger uncertainty set yields a more robust solution, which, in turn, raises the cost.
\begin{figure}[h!]
	\vspace{-0.2cm}
 \subfigure[Varying $\Gamma$ and $\rho$]{
	     \includegraphics[width=0.242\textwidth,height=0.10\textheight]{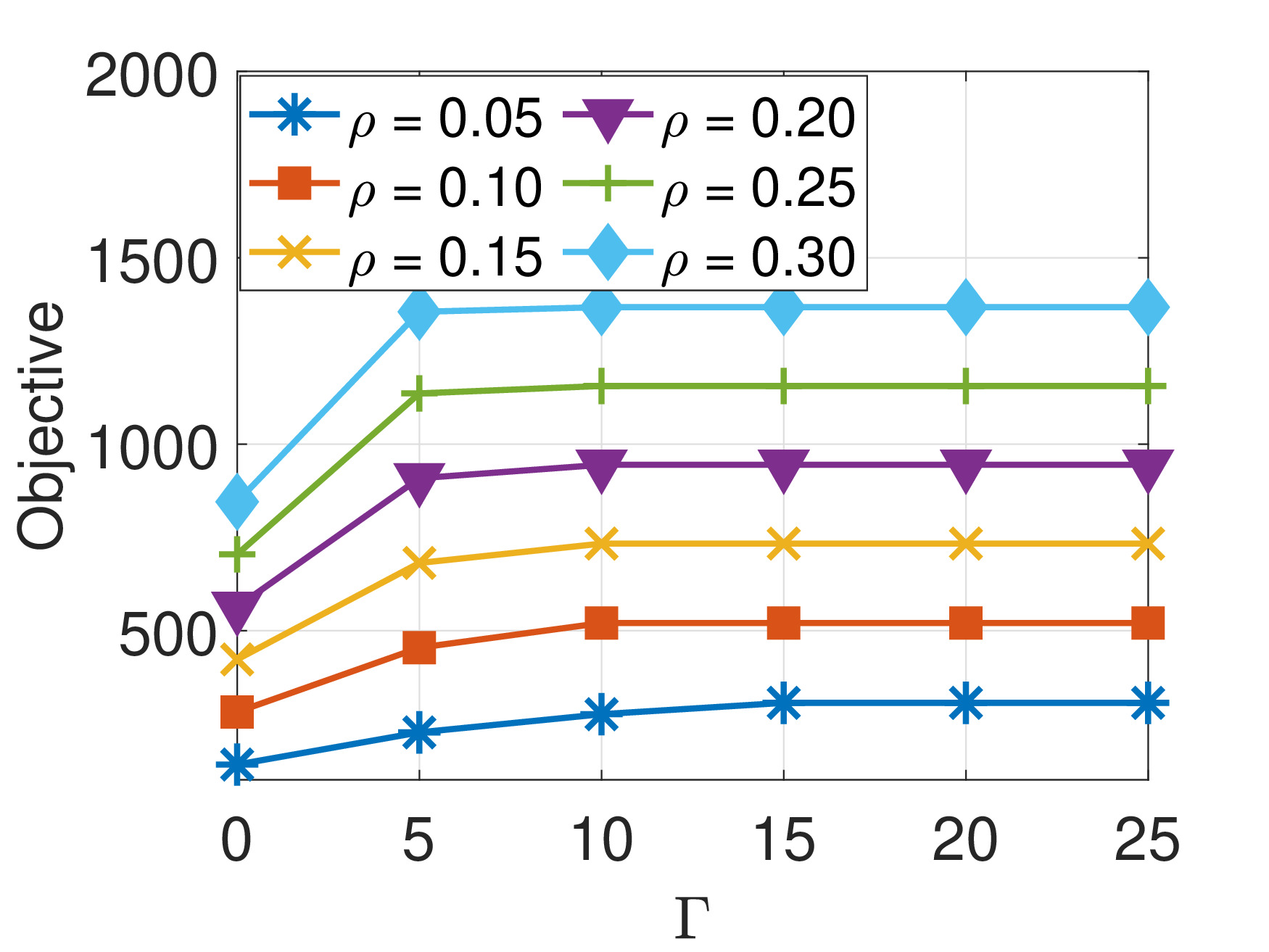}
	     \label{fig:Gamma_rho}
	}  \hspace*{-2.1em} 
	     \subfigure[Varying $\rho$ and $\bu$]{
	     \includegraphics[width=0.242\textwidth,height=0.10\textheight]{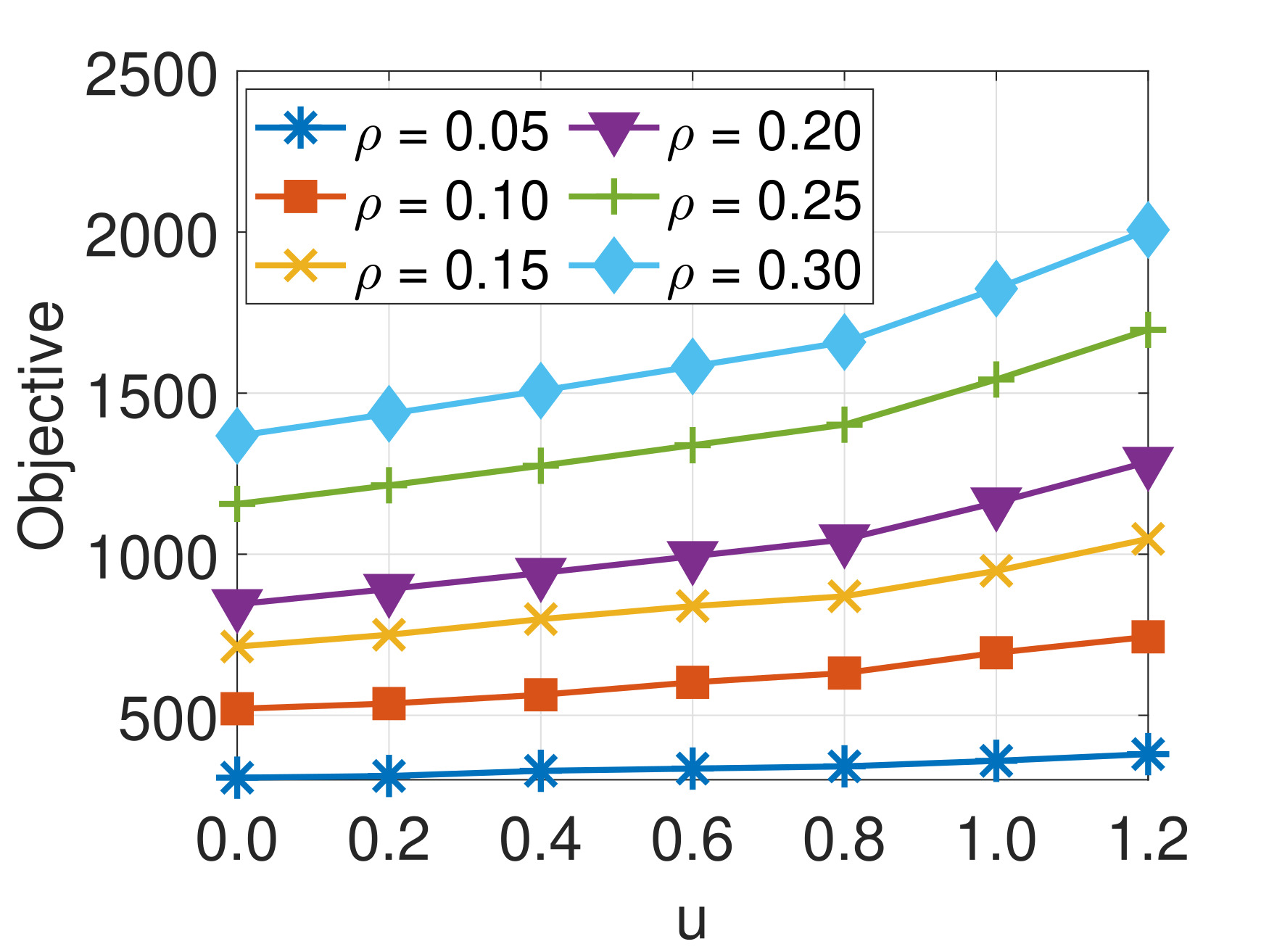}
	     \label{fig:u_Gamma}
	}
	\vspace{-0.2cm}
	    \caption{Impacts of the uncertainties on the system performance}
	    \vspace{-0.3cm}
\end{figure}

Figs.~\ref{fig:Gamma_gamma_psi_Cost}-\ref{fig:first_level_gamma_psi_payment} demonstrate the impacts of link hardening cost parameters (i.e., $\Psi$) and link hardening impact factors (i.e., $\gamma$)  on the optimal objective and hardening cost (i.e., the total payment for link hardening). The incremental impact ($\Delta \gamma$) between two adjacent link hardening levels is also considered, where a higher value of $\Delta \gamma$  indicates a greater impact of level-2 and level-3 hardening, thereby offering more incentives for the platform to select links for level-2 and level-3 hardening. From  $\mathcal{D}_2$, an increase in the link hardening impact results in lower link delays,  leading to a reduction in overall delay cost.  As $\Delta \gamma$ increases, so does its impact. 

Fig.~\ref{fig:Gamma_gamma_psi_Cost} shows that as $\Delta \gamma$ increases, the total cost decreases due to the significant reduction in link delay achieved through link hardening.   This finding is further supported by Figs.~\ref{fig:Gamma_Deltagamma_cost}--\ref{fig:Gamma_Deltagamma_payment}. In addition, Fig.~\ref{fig:Gamma_gamma_psi_payment} indicates that the link hardening payment initially increases with the link hardening cost scaling factor $\Psi$ and then decreases after a certain value of $\Psi$,  beyond which the hardening cost outweighs the benefit.

In Figs.~\ref{fig:first_level_gamma_psi_cost}--\ref{fig:first_level_gamma_psi_payment}, we set the incremental impact factor $\Delta \gamma$ to 0.3 and vary first-level impact factor ($\bm{\gamma}^{\sf 1} = \gamma_{i,j}^{\sf 1}, \forall i,j$) from $0.1$ to $0.4$. A larger value of $\gamma^{\sf 1}$ leads to a lower link delay. 
 Fig.~\ref{fig:first_level_gamma_psi_cost} reveals that 
the total cost decreases as $\gamma^{\sf 1}$ increases,  while Fig.~\ref{fig:first_level_gamma_psi_payment} implies the platform invests more in link hardening as $\gamma^{\sf 1}$ increases. 
Finally, we vary  $\Delta \gamma$ and the uncertainty budget $\Gamma$ to study their impact on the total cost and the link hardening payment in 
Figs.~\ref{fig:Gamma_Deltagamma_cost}--\ref{fig:Gamma_Deltagamma_payment}. It can be seen that an increase in $\Gamma$ leads to a higher cost due to the larger uncertainty set. Moreover, Fig.~\ref{fig:Gamma_Deltagamma_payment} suggests that the platform is willing to invest more money in link hardening as $\Gamma$ increases,  indicating a higher level of uncertainty.

\begin{figure}[h!]
\vspace{-0.2cm}
\label{Fig:uncertainty_impact_2}
	    \subfigure[Cost, varying  $\Psi$ and $\Delta \gamma$]{
	     \includegraphics[width=0.242\textwidth,height=0.10\textheight]{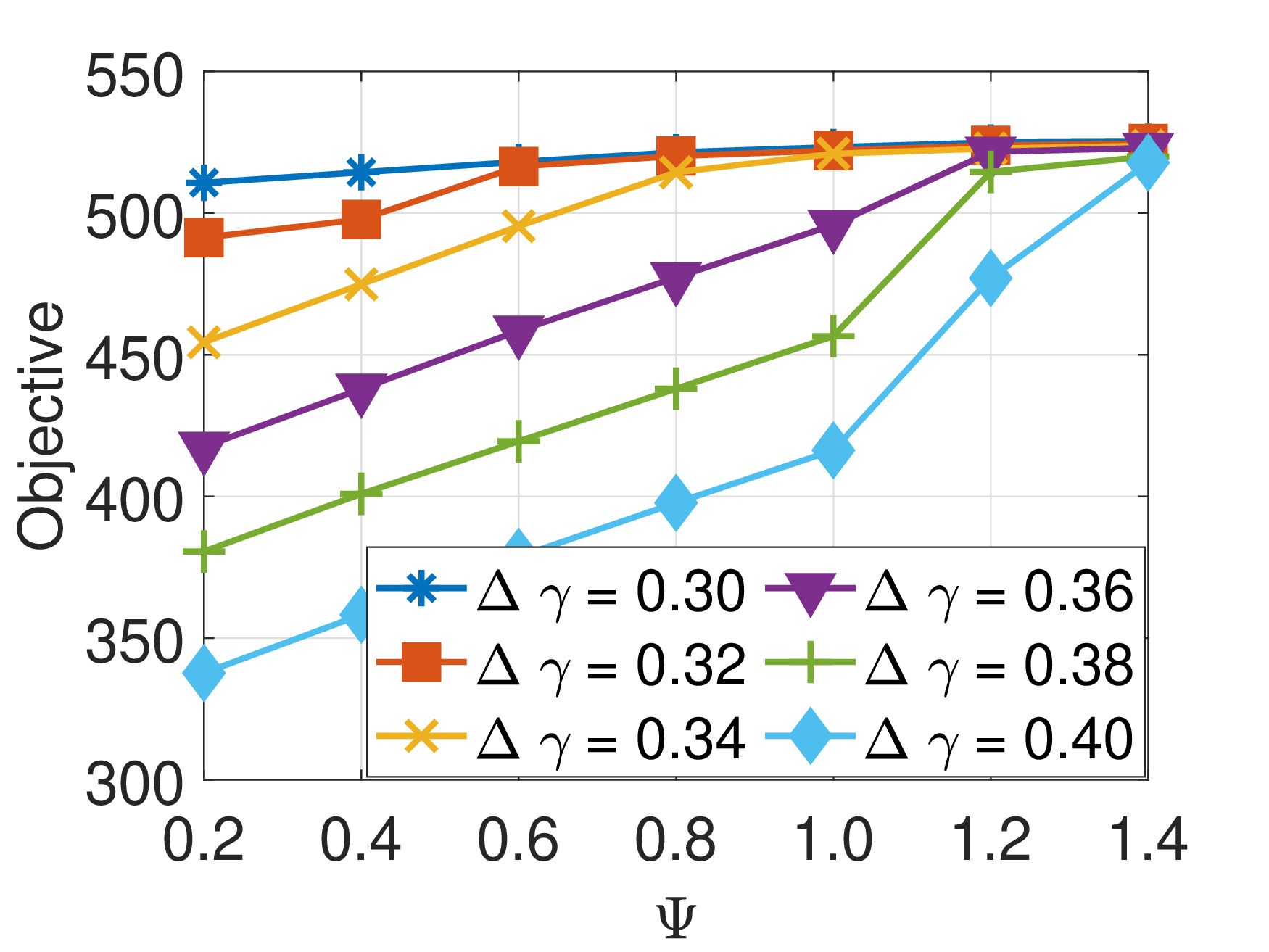}
	     \label{fig:Gamma_gamma_psi_Cost}
	} \hspace*{-1.9em} 
	    \subfigure[Payment, varying $\Psi$ and $\Delta \gamma$]{	
        \includegraphics[width=0.242\textwidth,height=0.10\textheight]{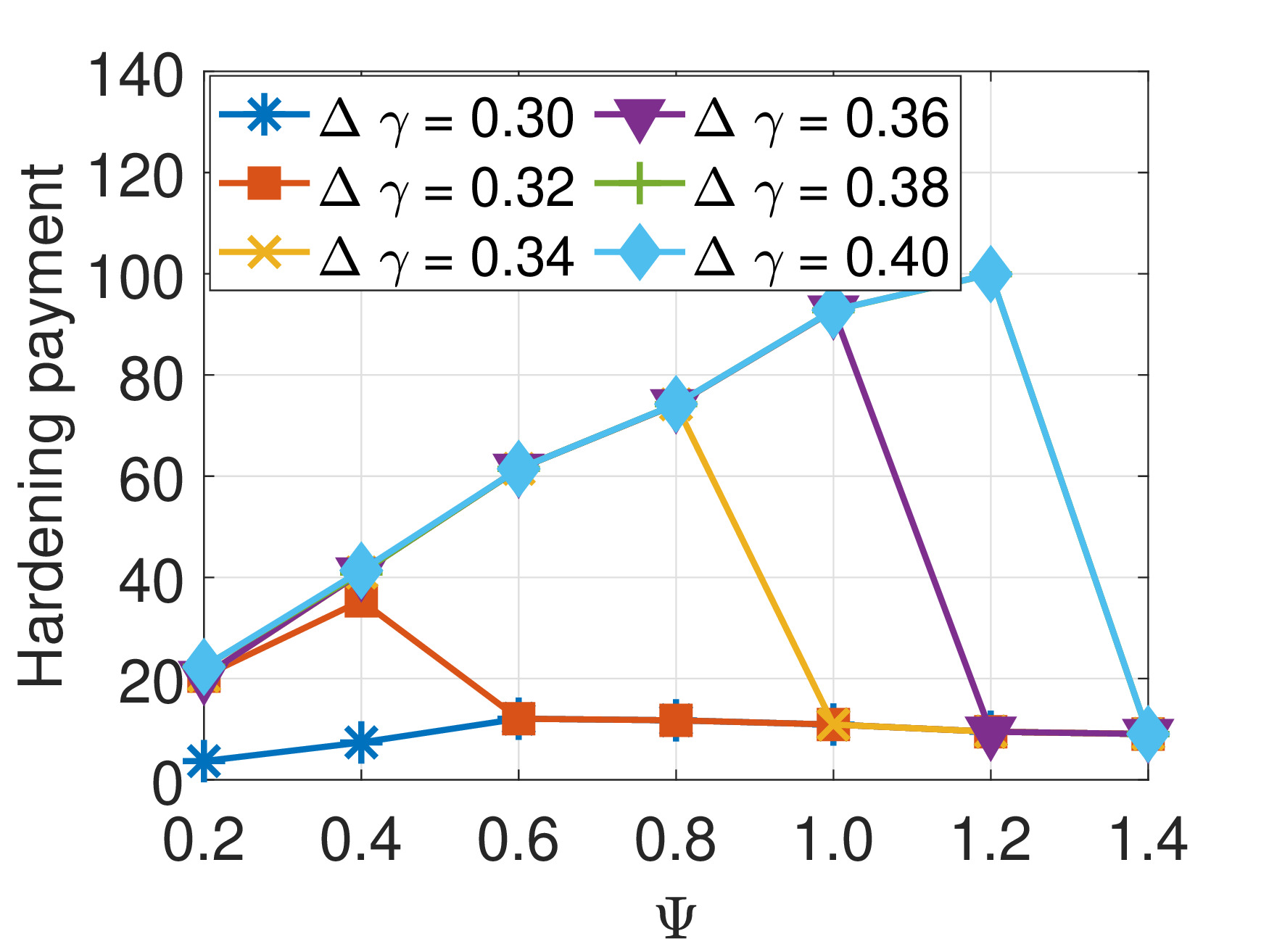}
	    \label{fig:Gamma_gamma_psi_payment}
	}   
	\subfigure[Cost, varying  $\gamma^{\sf 1}$ and $\Psi$]{
	     \includegraphics[width=0.242\textwidth,height=0.10\textheight]{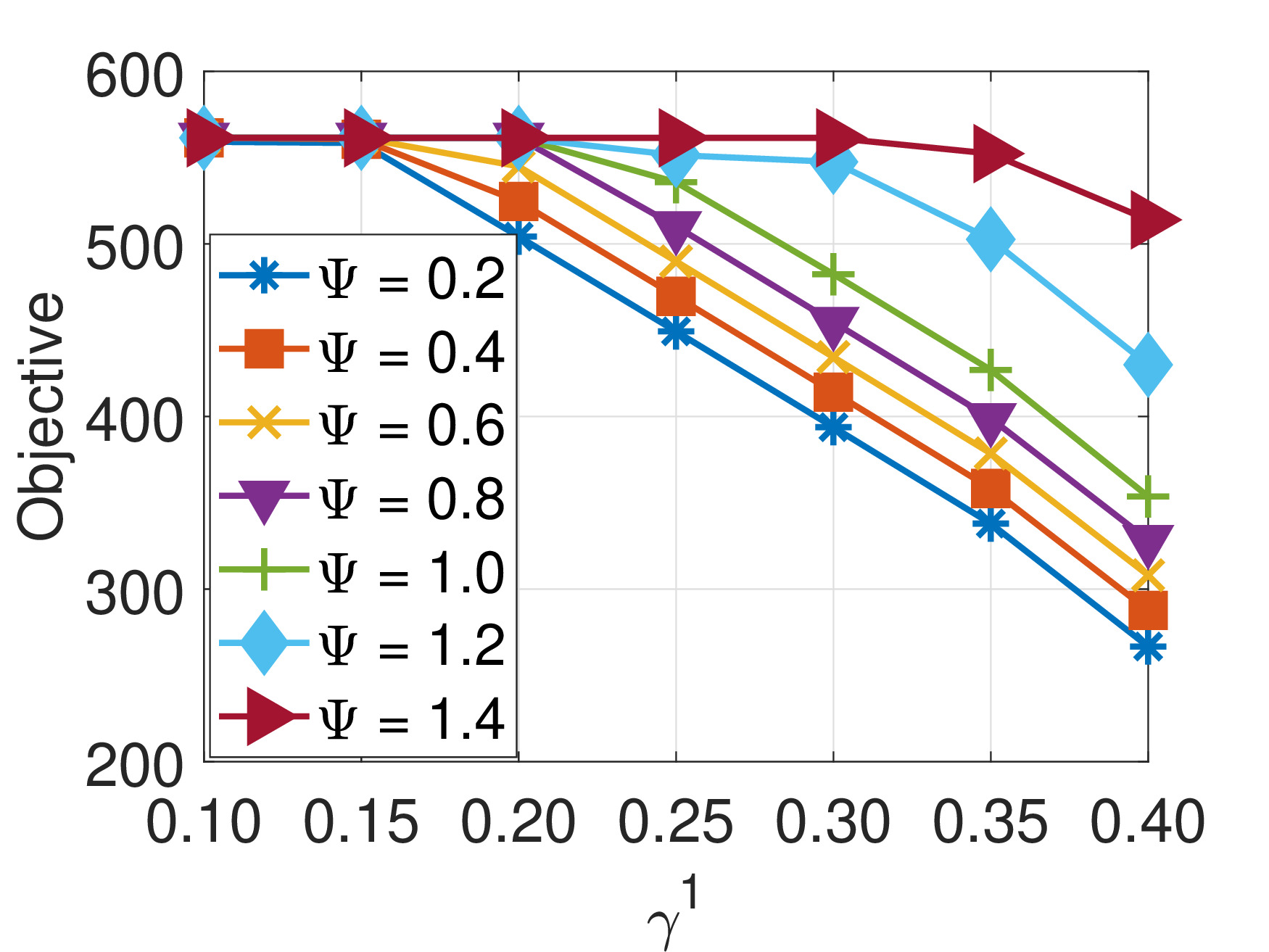}
	     \label{fig:first_level_gamma_psi_cost}
	} \hspace*{-1.9em} 
	    \subfigure[Payment, varying $\gamma^{\sf 1}$ and $\Psi$]{	
        \includegraphics[width=0.242\textwidth,height=0.10\textheight]{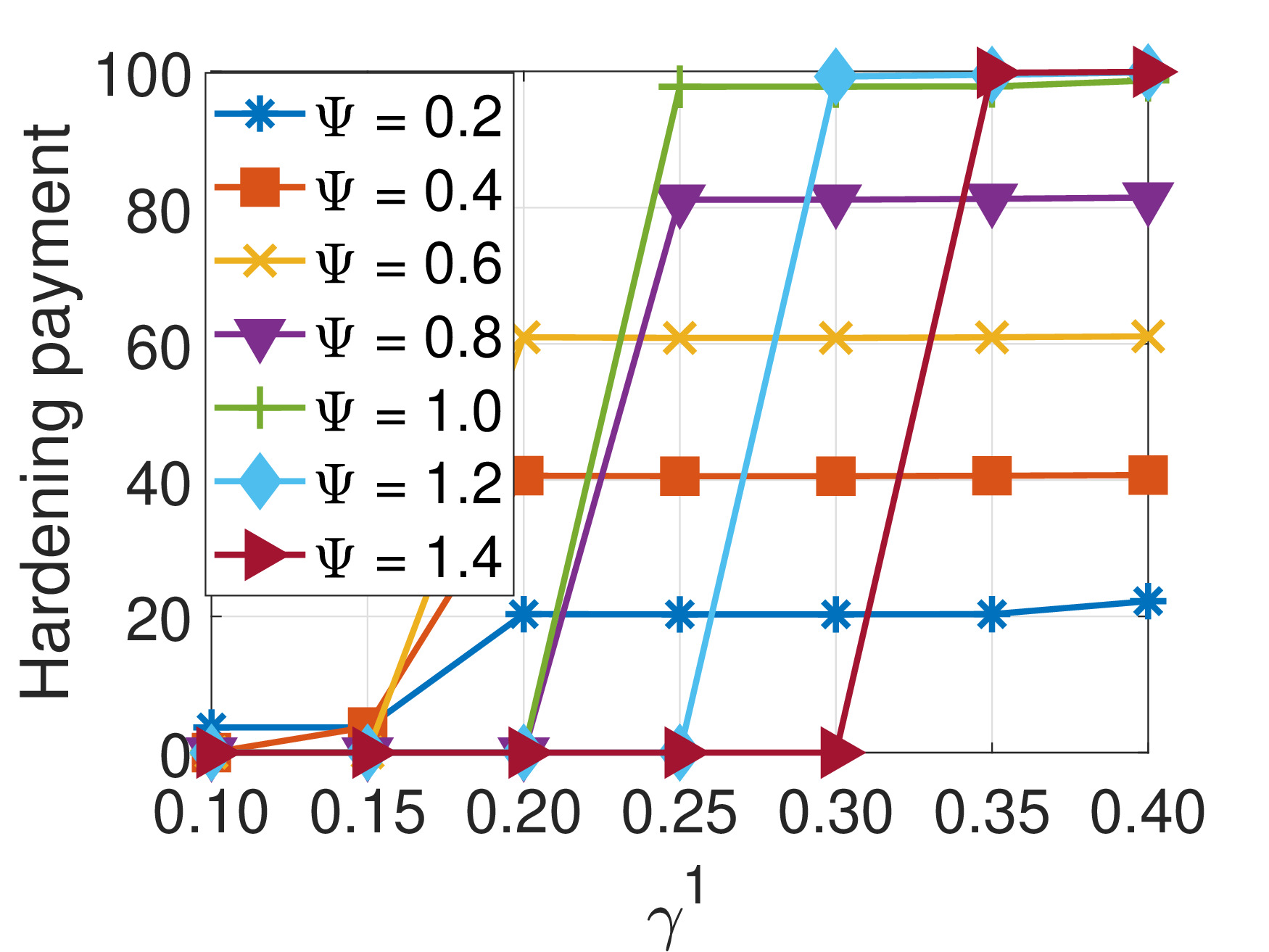}
	    \label{fig:first_level_gamma_psi_payment}
	} 
	 \subfigure[Cost, varying  $\Gamma$ and $\Delta \gamma$]{
	     \includegraphics[width=0.242\textwidth,height=0.10\textheight]{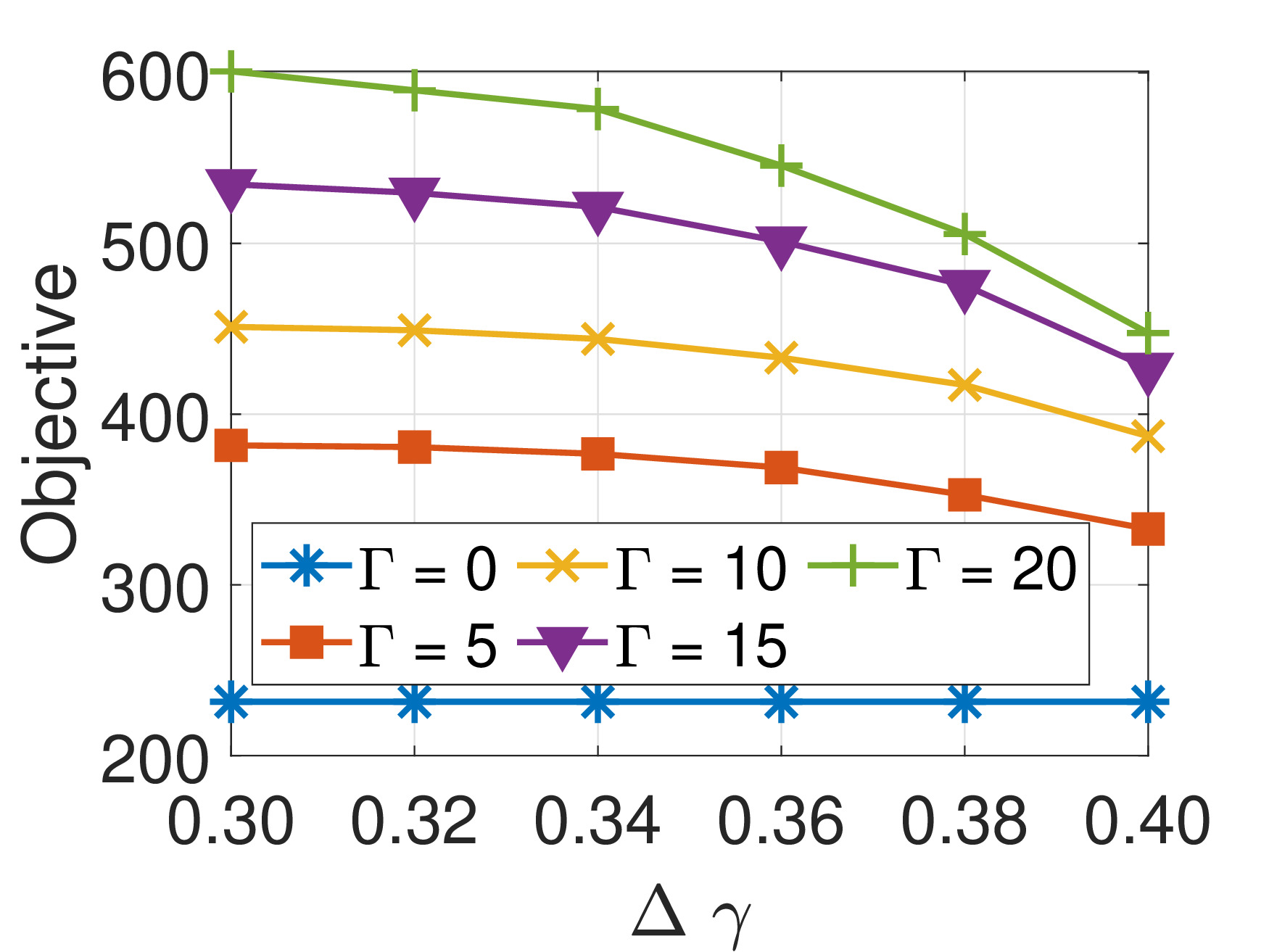}
	     \label{fig:Gamma_Deltagamma_cost}
	}    \hspace*{-1.9em} 
 	    \subfigure[Payment, varying $\Gamma$ and $\Delta \gamma$]{	 
      \includegraphics[width=0.242\textwidth,height=0.10\textheight]{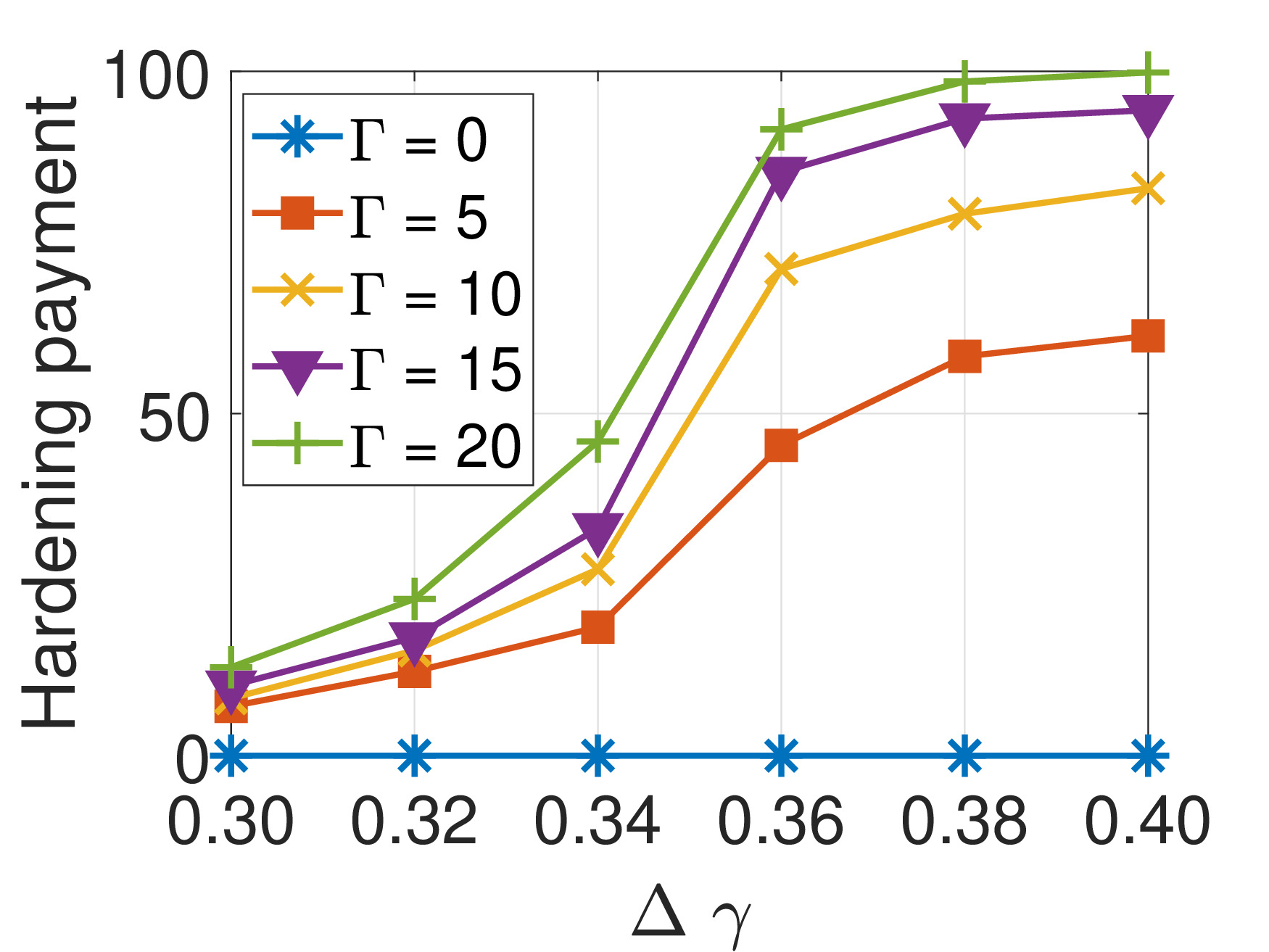}
	     \label{fig:Gamma_Deltagamma_payment}
	}   
	\vspace{-0.2cm}
	    \caption{Impacts  of the uncertainty set}
	    \vspace{-0.4cm}
\end{figure}

\subsubsection{Impact of network size}
Figs.~\ref{fig:ENs}-\ref{fig:APs} illustrate how the system size affects the optimal solution.  As presented in Fig. \ref{fig:ENs}, given a fixed budget, the total cost decreases when the number of ENs increases because 
the availability of more edge resources enables the platform to efficiently distribute the workload, ensuring that user requests are served by ENs located closer to them. This improved proximity helps to reduce delay and unmet demand penalties, thus enhancing service quality. In addition, the cost decreases rapidly at first, then more gradually later. This implies that the platform does not require too many ENs to maintain service quality, as unmet demands and delays decrease with the availability of more edge resources. On the other hand, Fig. \ref{fig:APs} reveals that, for the same set of ENs, the total cost increases as the number of areas $I$ increases, which translates to a growth in resource demand. This increase in demand generally leads to  higher delays and more unmet demand. This situation may reduce the platform's flexibility in resource allocation, potentially leading to unmet demand penalties from delay and resource capacity constraint violations. In contrast, increasing ENs with a fixed number of APs reduces backhaul pressure and provides more flexibility in allocating services to nodes, lowering total cost.
\begin{figure}[h!]
\vspace{-0.2cm}
		 \subfigure[$I$ = 20, varying $J$]{
	     \includegraphics[width=0.245\textwidth,height=0.10\textheight]{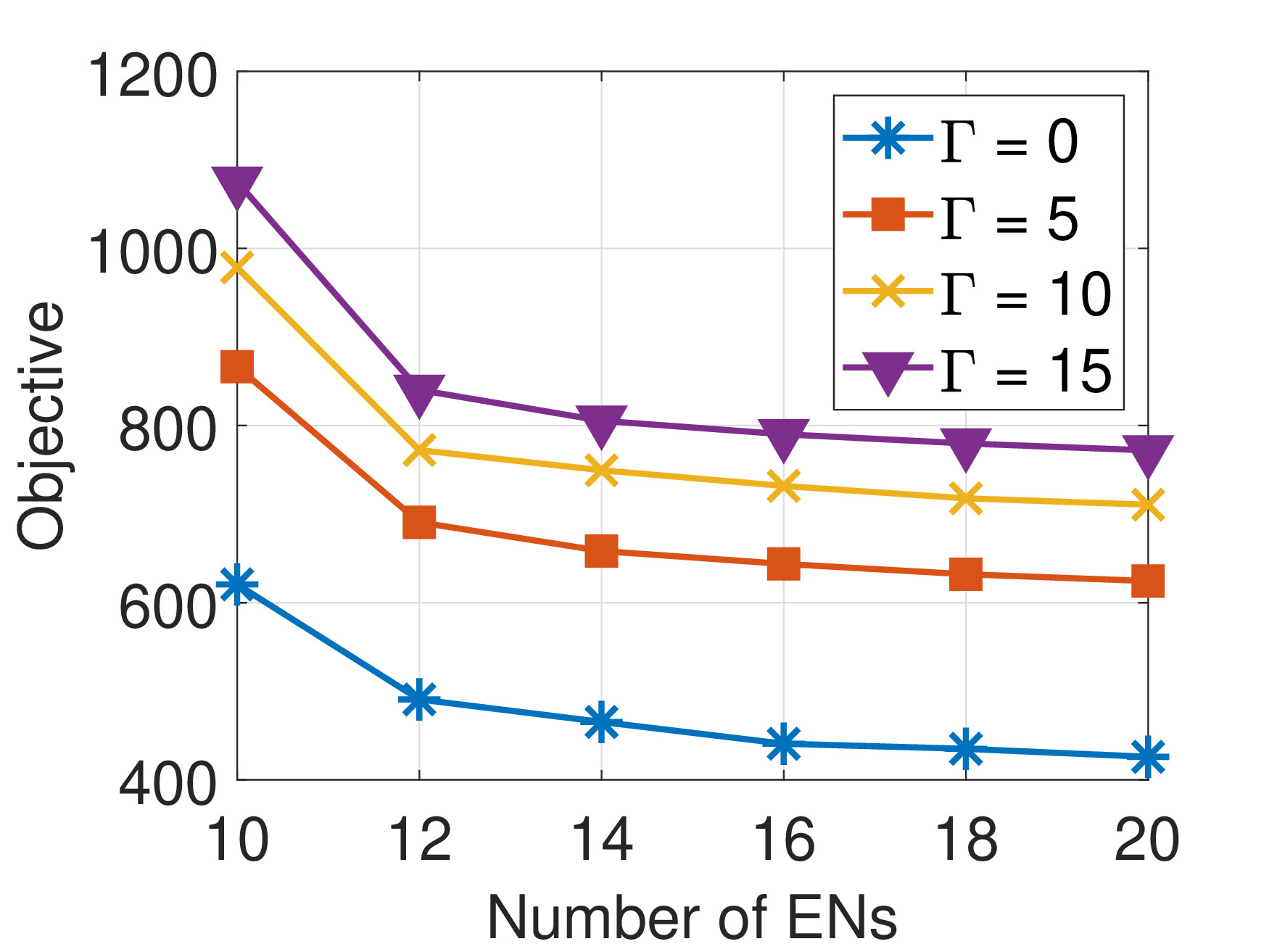}
	     %\caption{}
	     \label{fig:ENs}
	}   \hspace*{-2.1em} 
	    \subfigure[$J$ = 10, varying $I$]{
	     \includegraphics[width=0.245\textwidth,height=0.10\textheight]{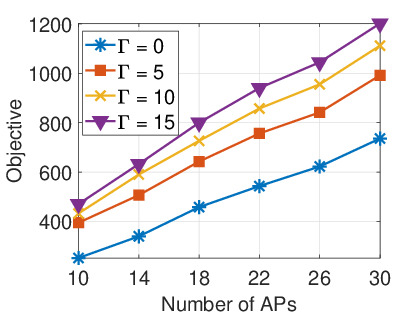}
	     %\caption{}
	     \label{fig:APs}
	} 
	\vspace{-0.2cm}
	    \caption{Impacts of the system size}
	    \vspace{-0.2cm}
\end{figure}

\subsubsection{Running time comparison between \textit{RDDU} and \textit{e-RDDU}} 
In the simulation, we set the optimality gap to   $0.2 \%$ and then compute the average running time over $100$ randomly generated problem instances for each problem size. Table \ref{tab: time_table} provides a comparison of the computational time taken by \textit{RDDU} and \textit{e-RDDU} for various problem sizes and levels of uncertainty. For smaller and medium-sized networks, both algorithms are capable of producing optimal solutions within a reasonable time. However, as the system size grows larger, the advantage of \textit{e-RDDU} over \textit{RDDU} becomes more pronounced due to the former's ability to generate fewer constraints when linearizing bilinear terms. It is worth emphasizing that \textit{the underlying problem is a robust planning problem that does not necessitate real-time computation}.

\begin{table}[t!]
\centering
\begin{tabular}{|l|l|l|l|}
\hline
\textbf{Network size}           & $\Gamma$ & \textit{RDDU}  (seconds) & \textit{e-RDDU} (seconds)   \\ \hline
% \multirow{3}{*}{I = 10; J = 5}  & 5     & 4.08  & 2.12    \\ \cline{2-4} 
%                                 & 10    & 2.74  & 2.31    \\ \cline{2-4} 
%                                 & 15    & 3.03  & 2.52   \\ \hline
\multirow{3}{*}{I = 10; J = 10} & 5     & 6.51 & 5.08    \\ \cline{2-4}
                                & 10    & 9.49  & 8.23   \\ \cline{2-4} 
                                & 15    & 10.91  & 8.35  \\ \hline
\multirow{3}{*}{I = 20; J = 10} & 5     & 32.71 & 24.84  \\ \cline{2-4}
                                & 10    & 34.70 & 32.55  \\ \cline{2-4} 
                                & 15    & 34.52 & 32.75   \\ \hline
\multirow{3}{*}{I = 30; J = 10} & 5     &  476.83  & 254.38 \\ \cline{2-4}
                                & 10    &  873.14   & 525.61  \\ \cline{2-4}
                                & 15    &  1152.22 & 760.06  \\ \hline
\multirow{3}{*}{I = 20; J = 20} & 5     &  556.31 & 422.79  \\ \cline{2-4}
                                & 10    & 1997.25 & 976.34  \\ \cline{2-4}
                                & 15    & 2205.8 & 1458.42      \\ \hline
\multirow{3}{*}{I = 30; J = 20} & 5     & 907.69 & 525.21  \\ \cline{2-4}
                                & 10    & 1243.51 &  875.43   \\ \cline{2-4}
                                & 15    & 2779.62 & 1283.71  \\ \hline
\multirow{3}{*}{I = 50; J = 20} & 5     & 3593.81 & 1260.50  \\ \cline{2-4}
                                & 10    & 4838.93  & 1908.43  \\ \cline{2-4}
                                & 15    & 11516.25 & 3579.12 \\ \hline
\end{tabular}
\caption{\revtwo{Runtime comparison}}
\label{tab: time_table}
\vspace{-0.4cm}
\end{table}

\subsection{Performance Comparison}
\label{percom}
In this section, we compare the performance of the proposed \textit{RDDU} model with the following benchmarks:

\begin{itemize}[leftmargin=0pt]
    \item[] {\scriptsize $\bullet$} \textit{NH}: There is \textit{no  hardening} at optimality in the deterministic model presented in Section \ref{deterministic}. The platform minimizes only the workload allocation cost in \eqref{eq-NOobj}. 
    \item[] {\scriptsize $\bullet$}  \textit{RAND}: We consider a \textit{randomized} link hardening scheme where the platform randomly selects a subset of links for hardening until exhausting the budget $B$. Specifically, the platform has no preferences on which link and which hardening level to be chosen. Thus, it can successively select links for hardening at random until exhausting the budget. Note that there is at most one hardening level chosen for each link. Thus, if a link has been selected at an iteration, it will be removed from the set of available links for hardening in future iterations. 
    \item[] {\scriptsize $\bullet$} \textit{SDDU}: The platform optimizes the link hardening decision by solving a stochastic model under DDU, where the uncertainty link delay is modeled through a set of decision-dependent scenarios. The SDDU model is presented in \textit{Appendix \ref{so_ddu}}. 
\end{itemize}

While the goal of the platform is to find the optimal link-hardening decision ($\bt$) before the disclosure of uncertainty,  it is not appropriate to evaluate the four schemes based on their planning costs. For instance,  the deterministic model (i.e., \textit{NH}) produces the lowest cost since it does not implement any link hardening. Therefore,  we compare the actual performance of the proposed scheme and the benchmarks instead.

Specifically, each scheme provides a hardening solution $\bt$ in the planning stage. Given the hardening decision and the actual realization of the uncertainties, the platform can re-optimize the workload allocation decision to minimize the cost. The main challenge in evaluating the benchmark schemes lies in the fact that they do not consider the interdependence of the workload allocation decision and the realized link delay. In particular, the link delay parameters are needed to optimize the workload allocation while we only know the actual link delay after making the workload allocation decisions. 

In the following, we describe how to compute the actual cost for each scheme. First, let $\Tilde{\bt}$ denote the link-hardening decision produced by each scheme, which is then used to generate the actual link delays. Subsequently, for each scheme, we obtain the corresponding uncertainty set based on their planning decision $\Tilde{\bt}$ as follows:
\begin{align}
\label{delay_gen}
&\mathcal{D}_2^a (\hat{d},\bar{d},t,x,\Gamma_2) \! := \! \bigg\{ d_{i,j}^a : ~  0 \leq g_{i,j} \leq 1 - \sum_{r} \gamma_{i,j}^{r} \Tilde{t}_{i,j}^{r}, \nonumber \\
&d_{i,j}^a \!=\! \bar{d}_{i,j} \!+\! \hat{d}_{i,j} ( g_{i,j} + u_{i,j} x_{i,j}), \forall i,j;~\!\sum_{i,j} g_{i,j} \leq \Gamma_2  %\label{actual_delay}\\
 %\label{harden_impact} 
    \bigg\}.
\end{align}

In our experimental setup,   we generate  $1000$ scenarios for each scheme to model the actual  link delays ($d_{i,j}^a$) by randomly generating $g_{i,j}$  satisfying the conditions in \eqref{delay_gen}. Note that for each scenario, the delay $d_{i,j}^a$ is a function of $x_{i,j}$. 
The platform then solves the following actual workload allocation problem:
\vspace{-0.2cm}
\begin{subequations}
\label{Actual_model}
\begin{align}
    \min_{\bw \in \mathbb{Z}_{+}^{I}, \bx \in \mathbb{Z}_{+}^{I \times J}}  \quad & \rho \sum_{i,j} d_{i,j}^a x_{i,j} + \sum_{i} s_i w_i\\
      \text{s.t.} \quad \quad \quad & \eqref{resource_cap}-\eqref{QoS} \\
    & d_{i,j}^a \!=\! \bar{d}_{i,j} \!+\! \hat{d}_{i,j} ( g_{i,j} + u_{i,j} x_{i,j}), \forall i,j \\
    &  \sum_{j} d_{i,j}^a \frac{x_{i,j}}{\lambda_i} \leq \Delta_i, \forall i.
\end{align} 
\end{subequations}
We consider two cases in our analysis: 
\begin{itemize}
    \item \textbf{Case 1}: the workload does not affect the link delays (i.e., $u_{i,j} = 0,~ \forall i, j$). Then, \eqref{Actual_model} becomes an ILP. 
    \item \textbf{Case 2}: the workload affects the link delays, resulting in bilinear terms $x_{i,j} x_{i,j}$ in problem \eqref{Actual_model}. We can employ the linearization techniques in \eqref{xxlin}-\eqref{RO_DDU_McCormick2} to linearize these bilinear terms. Thus, the actual workload allocation problem \eqref{Actual_model}  can be rewritten as an ILP. 
\end{itemize}

The actual total cost is the sum of the link hardening cost and actual workload allocation cost, expressed as:
\begin{align}
\label{Actual_cost}
    \mathcal{C}^{\sf a} = \sum_{i,j,r} h_{i,j}^{r} \Tilde{t}_{i,j}^{r}  + \rho \sum_{i,j} d_{i,j}^a x_{i,j}^{\sf a} + \sum_{i} s_i w_{i}^{\sf a},
\end{align}
where ($x^a, w^a$) is the optimal solution to problem \eqref{Actual_model}. 

\begin{figure}[t!]
\centering
		 \subfigure[Case 1]{
	     \includegraphics[width=0.245\textwidth,height=0.10\textheight]{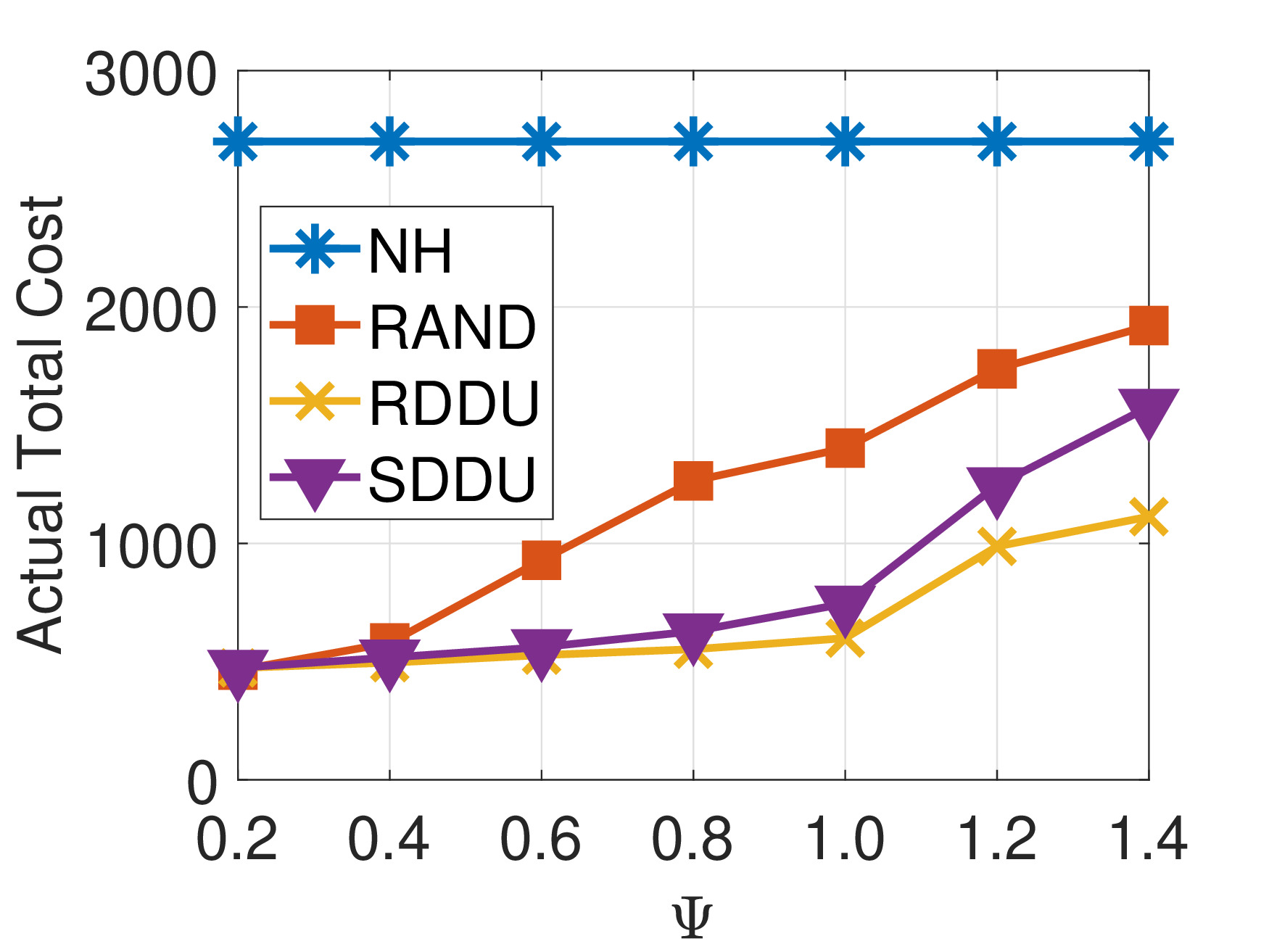}
	     %\caption{}
	     \label{fig:psi_model_case1_avg}
	}   \hspace*{-2.1em}
		\subfigure[Case 2]{
	     \includegraphics[width=0.245\textwidth,height=0.10\textheight]{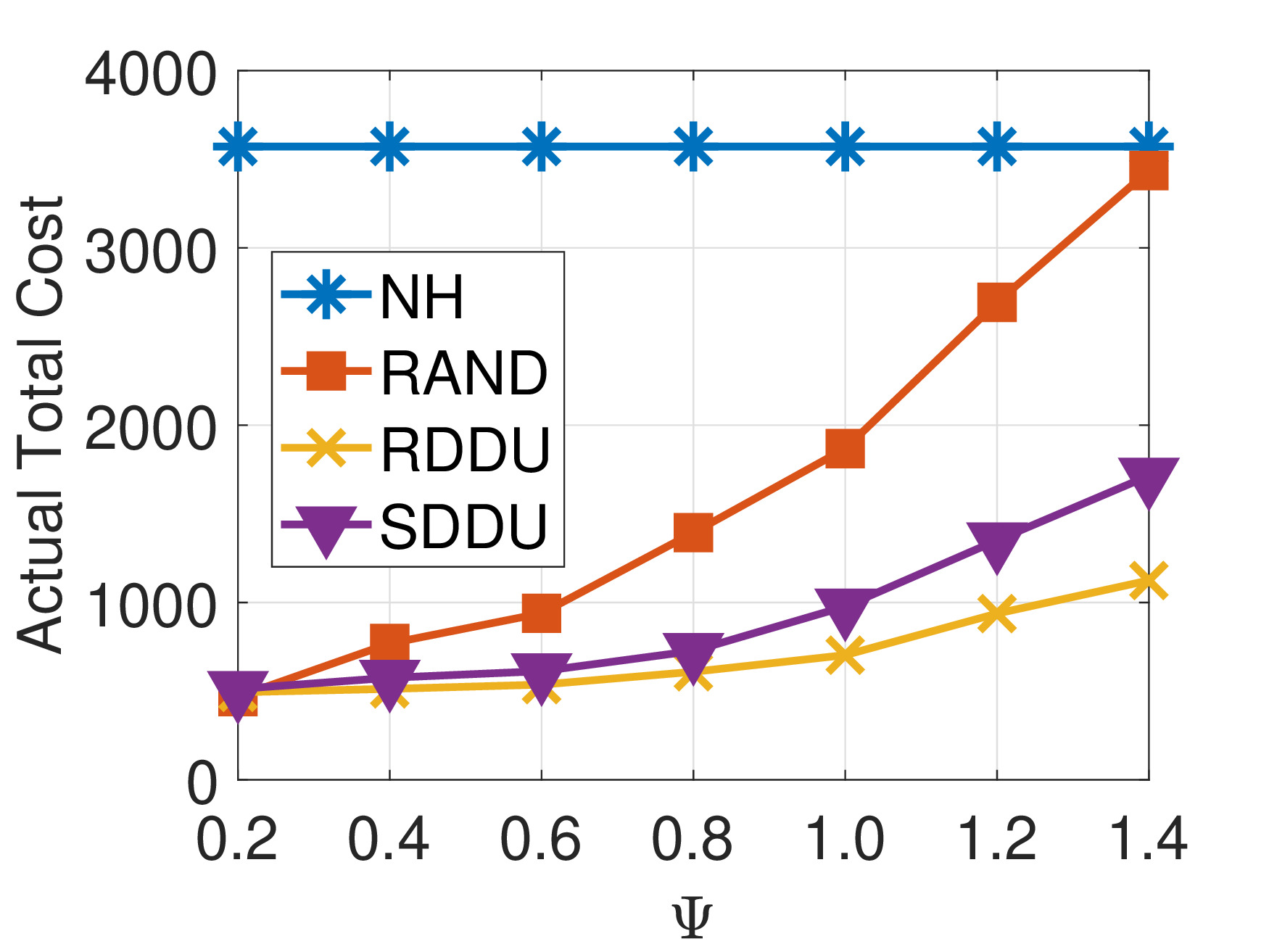}
	     %\caption{}
	     \label{fig:psi_model_case2_avg}
	}   
    	\subfigure[Case 1]{
	    \includegraphics[width=0.245\textwidth,height=0.10\textheight]{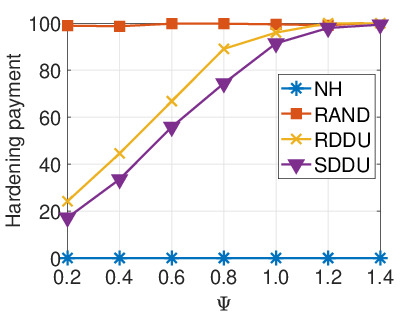}
        \label{fig:model_payment_case1}
    }    \hspace*{-2.1em} 
	    \subfigure[Case 2]{
	     \includegraphics[width=0.245\textwidth,height=0.10\textheight]{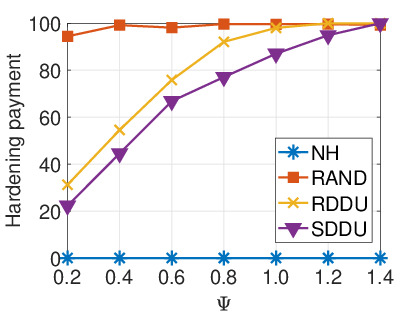}
	     \label{fig:model_payment_case2}
	}
	    
     \vspace{-0.2cm}
     \caption{Performance comparison, varying $\Psi$}
	    \vspace{-0.4cm}
\end{figure}

The evaluation and comparison of the four schemes are conducted based on their average actual costs over the generated scenarios.
As shown in
Figs.~\ref{fig:psi_model_case1_avg}--\ref{fig:psi_model_case2_avg}, the proposed \textit{RDDU} scheme  outperforms the other schemes significantly in terms of the actual total cost. Both $RDDU$ and $SDDU$ take into account the dependence of  uncertainty on the decisions when making the planning decisions, giving them an advantage over the other two schemes. 
However, SDDU's costs are higher than RDDU's since SDDU only considers a set of scenarios within the uncertainty set $\mathcal{D}_2^a$ while $RDDU$ is robust against the whole set. Furthermore, the performance of $SDDU$ depends on the knowledge of the distribution of the historical data.
The cost of the $NH$ scheme remains the same as $\Psi$ increases since \textit{NH} does not implement any hardening, regardless of the value of $\Psi$.  Figs.~\ref{fig:model_payment_case1}--\ref{fig:model_payment_case2} demonstrate that the hardening cost for $NH$ is consistently zero. The $RAND$ scheme, as defined,  always tends to exhaust the budget. Additionally, $RDDU$ suggests greater investment in link hardening than $SDDU$.

Figs.~\ref{fig:Varying_Delta}--\ref{fig:Varying_rho} 
compare the four schemes when varying the maximum delay threshold $\Delta$ and the delay penalty parameter $\rho$. It can be observed that for all  schemes, the total cost decreases as $\Delta$ increases (i.e., indicating greater relaxation in the delay requirement). On the other hand, as expected, the total cost increases with an increase in the delay penalty  $\rho$. Furthermore, the results demonstrate the superior performance of the proposed $RDDU$ scheme compared to the other schemes. 

\begin{figure}[t!]
		 \subfigure[Varying $\Delta$]{
	     \includegraphics[width=0.242\textwidth,height=0.10\textheight]{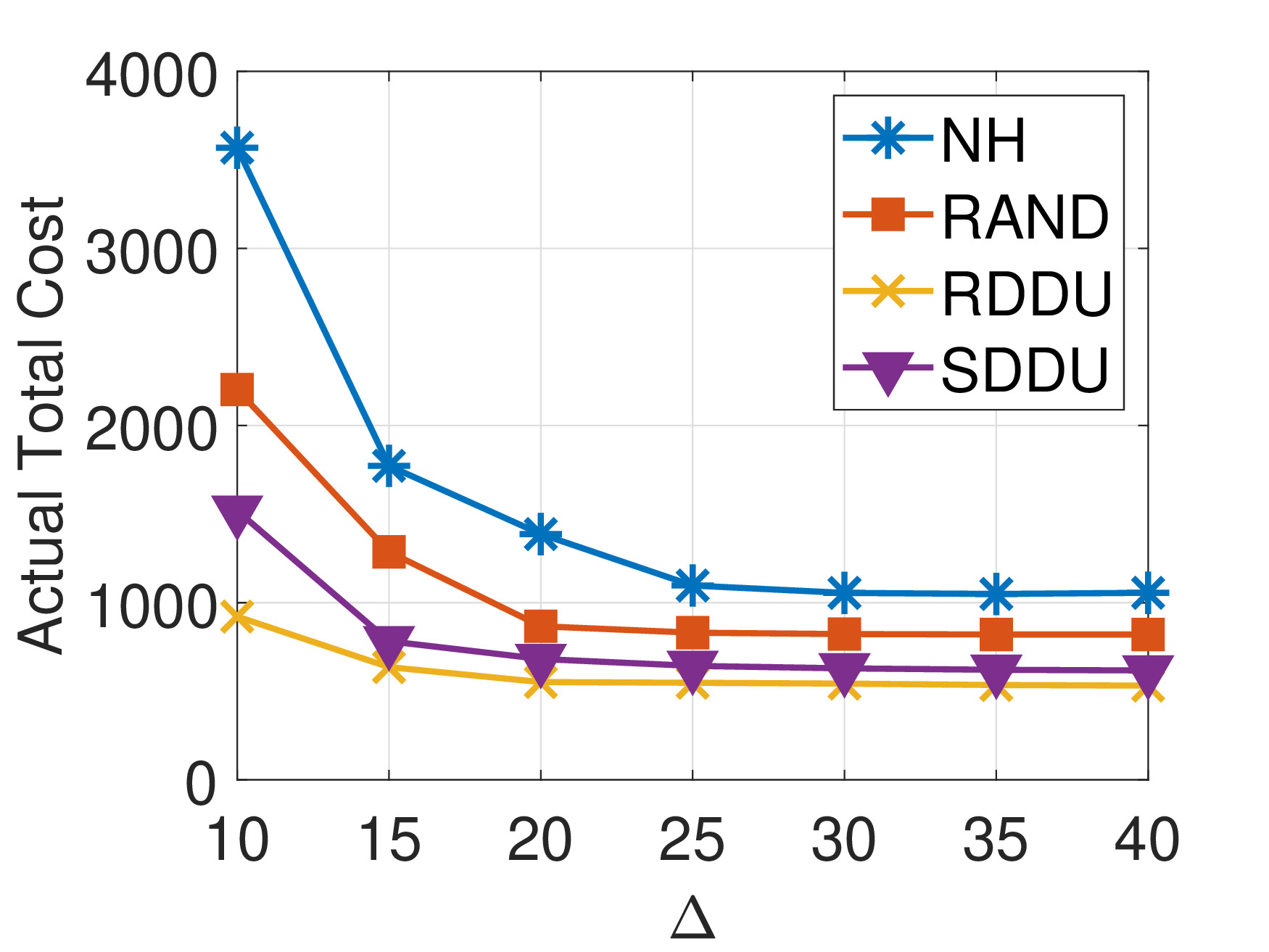}
	     \label{fig:Varying_Delta}
	}  \hspace*{-2.1em} 
	     \subfigure[Varying $\rho$]{
	     \includegraphics[width=0.242\textwidth,height=0.10\textheight]{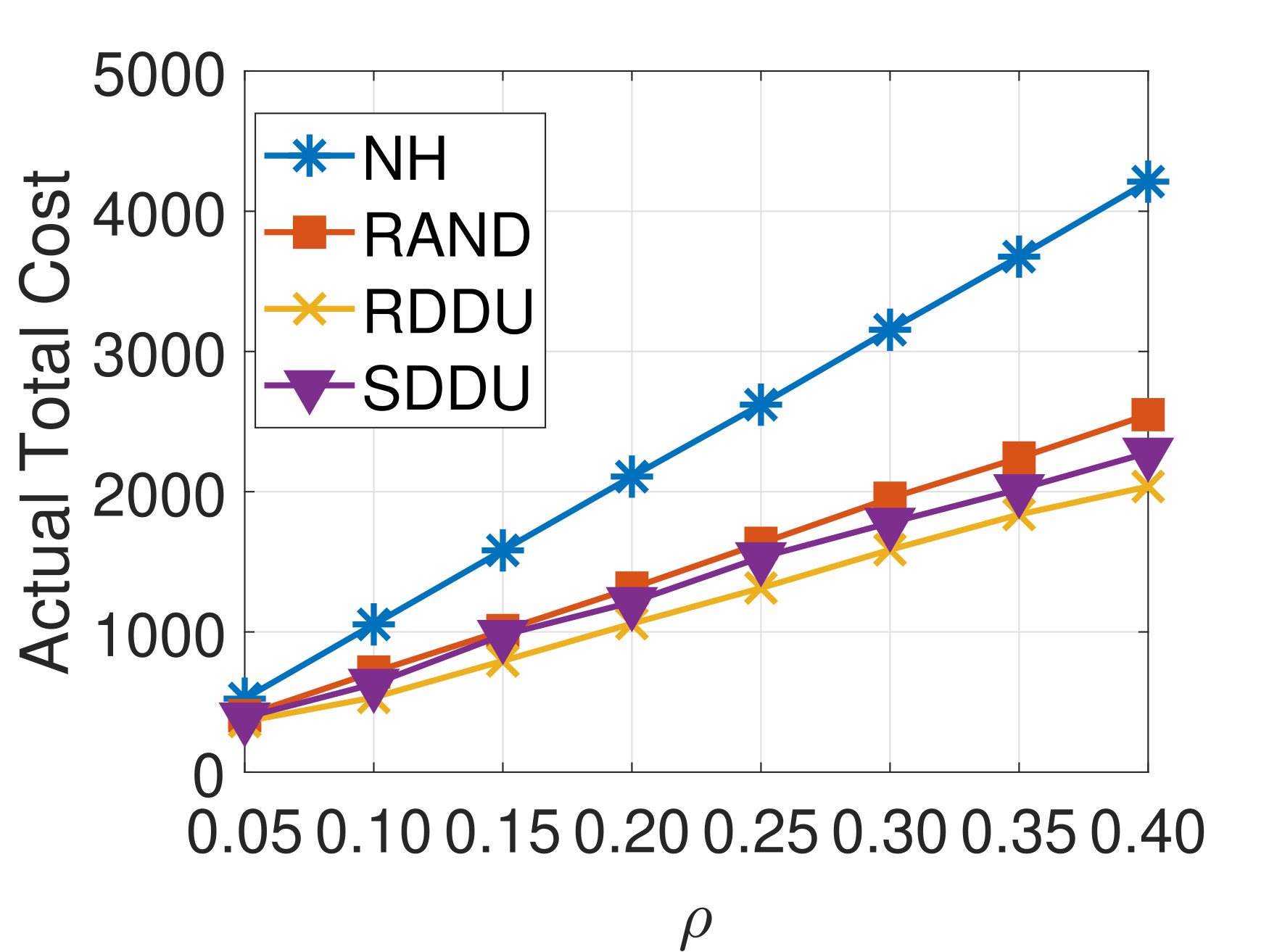}
	     \label{fig:Varying_rho}
	}
	\vspace{-0.2cm}
	    \caption{Performance comparison, varying $\Delta$ and $\rho$}
	    \vspace{-0.4cm}
\end{figure}

\section{Related Work}
\label{related_work}
\revtwo{This section provides an overview of existing research in edge computing, focusing on service placement, workload allocation under uncertainty, and network reliability, highlighting how our work extends and differs from previous studies.}

\textbf{\textit{Edge service placement and workload allocation:}} Extensive research has recently focused on optimizing edge service placement and workload allocation.
Pasteris \textit{et al.} \cite{He19} studied multi-service placement in EC to maximize rewards from serving user requests. Reference \cite{poularakis2019} jointly optimized service placement and request routing under resource capacity constraints. A reliability-aware dynamic service placement model based on the Markov decision process was proposed in \cite{Niyato20}  to minimize operational costs while maximizing accepted requests. The workload allocation problem was formulated as a MILP in \cite{Kherraf19}, aiming to maximize request acceptance while meeting strict QoS requirements.
References \cite{Duong_TON,Duong_TCC} proposed market equilibrium methods for fair and efficient edge resource allocation, and \cite{nguyen2023optimal} explored the data center operations considering renewable energy and batteries. However, \textit{most of these studies rely on deterministic optimization models}, making their performance heavily dependent on prediction accuracy, highlighting the need to incorporate uncertainty in EC systems.

\textbf{\textit{Edge resource management under uncertainties:}} Optimizing edge networks under uncertainty has garnered considerable attention in recent research. Reference \cite{niyato12} used stochastic programming to minimize edge resource provisioning costs amid demand and price uncertainties.  Ouyang \textit{et al.} \cite{ouyang19} introduced a bandit-based service placement scheme to jointly minimize transmission delay and service migration costs in an online setting. In \cite{wang22}, the authors studied a bandit-based dynamic service placement problem in a decentralized manner, aiming to maximize total reward while satisfying delay requirements. \cite{cheng23} proposed a multi-armed bandit approach for online pricing in heterogeneous edge resource allocation. A two-stage adaptive robust model was proposed in \cite{Duong_iot} to optimize edge service placement and sizing decisions under demand uncertainty. \textit{Unlike these works, our focus is on link delay uncertainty and how to effectively strengthen network links to proactively manage this uncertainty and ensure high service quality. This problem has been largely overlooked in existing EC literature.}

\textbf{\textit{Edge reliability and resiliency:}} There is growing interest in studying reliability and resiliency in networking and cloud/edge computing. In \cite{Modiano14,fhe22}, the authors proposed proactive network resource backup solutions to mitigate the impact of random network failures. Ito \textit{et al.} \cite{mito21} introduced a probabilistic protection model to minimize the backup capacity required by cloud providers in the event of random physical machine failures. In \cite{yhir20}, a backup network design scheme was presented to minimize the total required backup capacity with prior knowledge of the probabilistic distribution of required capacity. A resiliency-aware optimal edge service placement framework was proposed in \cite{Jiaming_TCC} to minimize operational costs under demand uncertainty and random node failures. \textit{Unlike these studies that primarily focus on protection against random failures, our research emphasizes proactive network link hardening to manage uncertainty and guarantee high service quality in EC.} While \cite{cheng2024two} proposed a decision-dependent distributionally RO model to capture the interdependence between EN placement decisions and uncertain demands, \textit{it overlooks uncertain network delays and their dependencies on link utilization.}

To the best of our knowledge, \textit{we are the first to examine link delay uncertainty in EC and incorporate it into a novel robust budget-constrained link hardening model}. In contrast to existing robust models that treat uncertainties as exogenous, our approach introduces an endogenous uncertainty set to accurately capture the interdependence between network link delay, link hardening, and workload allocation decisions. By leveraging this endogenous uncertainty set, our model empowers the platform to make decisions that optimize system performance while proactively controlling uncertainty levels. We also developed efficient algorithms for the formulated robust link hardening problem under endogenous delay uncertainty.

\section{Conclusion}
\label{conc}
This paper introduced a novel framework for robust link-hardening in EC to mitigate the impact of delay uncertainty on service quality. The main novelty lies in the integration of a DDU set, which effectively captures the interdependence between uncertain link delays and link hardening and workload allocation decisions. Furthermore, we devised two exact algorithms, \textit{RDDU} and \textit{e-RDDU}, to efficiently compute optimal solutions. The enhanced robust solution, \textit{e-RDDU}, significantly improves computational speed by eliminating redundant constraints without sacrificing solution optimality. Extensive numerical results underscore the advantages of accounting for endogenous uncertainties and demonstrate the superiority of the proposed model over benchmark schemes. The proposed model emphasizes the significance of considering DDU when making robust link-hardening decisions. \revtwo{In future work, we aim to leverage data-driven approaches, such as machine learning and deep learning techniques, to design fast heuristic algorithms that accelerate the computation of robust solutions for the proposed model. Our current methods, which generate exact solutions, can provide high-quality datasets for training, thereby enhancing predictive accuracy.}

\bibliographystyle{IEEEtran}
\bibliography{ref.bib}

 \cleardoublepage

\appendix

\subsection{NP-completeness for \textbf{RO-DDU} problem}
\label{appen:npcomplete}

The proposed single-stage RO problem with DDU sets can be summarized as follows:
\begin{subequations}
\begin{align}
    & \textbf{(RO-DDU)} \quad \min_{\bz,\bx} ~~ \bm{f}^{T} \bx + \bm{c}^{T} \bz \\
    & \text{s.t.} ~~ a_i^{T} \bx + \zeta_i^{T} \bz \leq b_i, ~ \forall  \zeta_i \in \mathcal{U}_i(\bz), ~ \forall i = 1,\dots,m.
\end{align} 
\end{subequations}
where $\bz \in \mathbb{R}^{n}$ and $\bx \in \mathbb{R}^{n}$ represents decision variables.

The subsequent theorem demonstrates that the RO problem featuring decision-dependent sets lacks a tractable reformulation, thereby deviating from the conventional RO problem. This occurs despite the existence of tractable robust counterparts for linear programs featuring polyhedral uncertainty sets.
\begin{theorem}
    The robust linear problem (\textbf{RO-DDU}) with a DDU set $\mathcal{U}(\bz) = \{\bm{\zeta} | \bm{A} \bm{\zeta} \leq \bv + \psi \bz\}$ is NP-complete, where $\bm{A}$ is a constant matrix, $\bv$ is a constant vector, and $\psi$ is an impact matrix.
\end{theorem}
\vspace{0.2cm}
\noindent
The proof proceeds through the following steps:

\vspace{0.2cm}
\noindent
1. Consider an instance of the 3-Satisfiability problem (\textbf{3-SAT}) defined over a set of literals $N = {1,2,\dots,n}$ and comprising $m$ clauses. The objective is to identify a solution $\bz \in {0,1}^{n}$ that satisfies the following condition for each of the $m$ clauses, and for all $i_1,i_2,i_3 \in \{1,\dots n\}$:
\begin{align}
z_{i_1} + z_{i_2} + (1 - z_{i_3}) \geq 1.
\end{align}

\vspace{0.2cm}
\noindent
2. Consider the following special instance of (RO-DDU), denoted as \textbf{RO-SAT}, where $\bz \in \mathbb{R}^{n}$, $\bx \in \mathbb{R}^{m}$, and $y \in \mathbb{R}$:
\begin{align}
\textbf{RO-SAT}: \min_{\bx,\bz,y \geq 0} \big\{ - y | y - \bm{\zeta}^{T} \bx \leq 0, \forall \bm{\zeta} \in \mathcal{U}(\bz), \\
\bx,\bz \leq \one, - \bx \leq -\one \big\} 
\end{align}
where $\one$ denotes the vector with all entries equal to $1$ and $\mathcal{U}(\bz) = \{ (\zeta_1,\zeta_2,\dots,\zeta_m) | \zeta_i \geq z_{i_1}, \zeta_i \geq z_{i_2}, \zeta_i \geq 1 - z_{i_3}, \zeta_i \leq 1 \} $ which encompasses the \textbf{3-SAT} problem.

\vspace{0.2cm}
\noindent 
3. According to Lemma~\ref{lem-3SATOptimal} (provided subsequently), the optimal value of (\textbf{RO-SAT}) is $-m$ if and only if a solution exists for the \textbf{3-SAT} problem.

\vspace{0.2cm}
\noindent
4. Given that problem (\textbf{RO-SAT}) is a special case of (\textbf{RO-DDU}) with a polyhedral set $\mathcal{U}(\bz)$, and considering the NP-completeness of the \textbf{3-SAT} problem, it can be deduced that problem (\textbf{RO-DDU}) is likewise NP-complete.

\vspace{0.2cm}
\noindent

\begin{lemma}\label{lem-3SATOptimal}
A feasible solution $\bz$ exists for the \textbf{3-SAT} problem if and only if problem (\textbf{RO-SAT}) attains an optimal value of at most $-m$.
\end{lemma}

\vspace{0.2cm}
\noindent 
\textit{\textbf{Proof. Sufficient condition:}} Assume that a feasible solution $\bz$ exists for the \textbf{3-SAT} problem. This implies that $\bz$ must adhere to the following constraint:
\begin{align}
    z_{i_1} + z_{i_2} + (1 - z_{i_3}) \geq 1, ~ \forall i = 1, \dots, m.
\end{align}
Given that $\bz \in \{0,1\}^{n}$, $\forall i$, at least one of $z_{i_1}$, $z_{i_2}$, or $1 - z_{i_3}$ must be equal to $1$. We then examine the uncertainty set:
\begin{align}
    \mathcal{U}(\bz) = \{ &(\zeta_1,\zeta_2,\dots,\zeta_m) | \zeta_i \geq z_{i_1}, \zeta_i \geq z_{i_2}, \nonumber\\
    &\zeta_i \geq 1 - z_{i_3}, \zeta_i \leq 1; ~ \forall i = 1, \dots, m \}. 
\end{align}
Given that at least one among $z_{i_1}$, $z_{i_2}$, or $1 - z_{i_3}$ equals $1$, it follows that $\zeta_i$ satisfies $\zeta_{i} \geq 1$. Since $\zeta_{i} \leq 1$, this implies that $\zeta_{i} = 1$ for all $i$ is the only point in $\mathcal{U}(\bz)$. Consequently, for this uncertainty set, the feasible solution is $\bz=1$, $\bx = 1$, and $y = m$. This results in the optimal solution $-y \leq -m$, or equivalently, $y \geq m$.

\vspace{0.2cm}
\noindent 
\textit{\textbf{Necessary condition:}} Assume (\textbf{RO-SAT}) achieves an optimal solution ($\bx^{*}$,$\bz^{*}$) with an objective value of $-y^{*} \leq -m$. The constraints in (\textbf{RO-SAT}) enforce $y^{*} - \bm{\zeta}^{T} \bx^{*} \leq 0$, implying $ \bm{\zeta}^{T} \bx^{*} \geq y^{*} \geq m, \forall \bm{\zeta} \in \mathcal{U}(\bz^{*})$. Moreover, these constraints necessitate  $y_i^{*} = 1$ for all $i$. Consequently, $\sum_{i = 1}^{m} \zeta_i > m, \forall \bm{\zeta} \in \mathcal{U}(\bz^{*})$. However, the uncertainty set construction limits $\zeta_i \leq 1$, resulting in a contradiction as $\sum_{i}^{m} \zeta_i \ngtr m$. Consequently,  $-y^{*} = -m$. Hence, $\bm{\zeta}^{T} \bx^{*} = m$ for all $\zeta \in \mathcal{U}(\bz^{*})$. This implies $\sum_{i = 1}^{m} \zeta_i = m$ for all $\bm{\zeta} \in \mathcal{U} (\bz^{*})$, leading to $\min_{\zeta \in \mathcal{U}(\bz^{*})} \sum_{i}^{m} \zeta_i = m$. However, since the uncertainty set dictates $\zeta_i \leq 1 \forall i$, the sum can only equal $m$ if $\zeta_i = 1$ for all $i$.

\vspace{0.2cm}
\noindent 
We are now prepared to demonstrate that for each $i$, at least one of $z_{i_1}^{*}$ or $z_{i_{2}}^{*}$ or $(1 - z_{i_3}^{*})$ equals $1$. Suppose this condition does not hold. This implies the existence of an $i$ such that $z_{i_{1}}^{*} \leq 1,z_{i_{2}}^{*} \leq 1$ and $1 - z_{i_{3}}^{*} \leq 1$. Consequently, we can define $\zeta_i^{'} = \max \{z_{i_1}^{*},z_{i_2}^{*},1 - z_{i_3}^{*} \}$, which is an element of the uncertainty set and  $\zeta_i^{'} < 1$. However, this contradicts the condition $\zeta_i = 1, \forall i$. Therefore, if $y^{*} = m$, then we can find a feasible solution for the \textbf{3-SAT} problem.

\noindent\textbf{Remark:} Even though problem (RO-DDU) is NP-complete, it can be reformulated into a bilinear or biconvex program, allowing for potential solutions through global optimization methods \cite{kolodziej2013global}. In cases involving binary decision variables $\bz$ impacting $\mathcal{U}(\bz)$, (RO-DDU) can be restructured as a  MILP utilizing linearization techniques \cite{BigM}.

\subsection{Stochastic model under decision-dependent uncertainty} %(\textit{SDDU})}
\label{so_ddu}
This section presents a stochastic model for edge network hardening under DDU, where the uncertainty link delay is modeled through a set of decision-dependent scenarios. The stochastic model assumes that link delay $d_{i,j}$ follows certain probability distributions, which can be obtained from historical data. Specifically, let $\xi_n \in \{\Xi\}_{n = 1}^{N}$ be the sample space of all uncertain realizations and $\bm{\pi}_n \in [0,1]$ be the probability of corresponding scenario $n$.  Here, $\bm{\pi}_n$ follows certain distribution and satisfies condition $\sum_{n=1}^{N} \pi_n = 1$. To capture the dependencies on decisions, we model this endogenous stochastic uncertainty set as follows: 
\begin{align}
\label{so_ddu_prob}
    d_{i,j}^{n} (\xi_n,\bt,\bx) = \bar{d}_{i,j} \!+\! \xi_{n} \hat{d}_{i,j} \bigg[1\!-\!\!  \sum_{r} \gamma_{i,j}^{r} t_{i,j}^{r} \!+\!  u_{i,j} x_{i,j} \bigg]\!, \forall i,j,n
\end{align}
where $\xi_n$ can follows truncated distribution from the interval $[0,1]$. Thus, the resulting actual delay $d_{i,j}^{n}$ follows the certain distribution with mean $\mu = \bar{d}_{i,j}$ and deviation $\sigma^2 = \hat{d}_{i,j} ( 1-  \sum_{r} \gamma_{i,j}^{r} t_{i,j}^{r} +  u_{i,j} x_{i,j})$. The platform considers scenario-dependent recourse decisions in this stochastic decision-dependent framework by treating the delay uncertainty as discrete scenarios with assigned probability. The objective of this \textit{SO-DDU} model is to optimize the expected cost over all scenarios, which can be expressed in \eqref{SO_obj}:
\begin{align}
    \!\!\min_{t,x,w}~  \sum_{i,j,r}\! h_{i,j}^{r} t_{i,j}^{r} \!\!+\! \mathbb{E}_{\mathbb{P}(\bt,\bx)} \!\bigg[ \rho\!\! \sum_{i,j} \!d_{i,j} (\xi_n,\bt,\bx) x_{i,j} \!\!+\!\!\! \sum_{i}\! s_{i} w_{i} \bigg]\!. \!\! \label{SO_obj}
\end{align}
Note that $\Xi$ is a discrete finite sample space. We will generate $N$ scenarios in the simulation based on chosen distribution from $[0,1]$. Using the linearity and convexity of expectation, we can express the expectation explicitly, shown in \eqref{so_obj}. Note that actual delay $d_{i,j}^{n}$ contains decisions ($\bx$, $\bm{t}$) for each scenario and two bilinear terms need to be linearized, i.e., $Z_{i,j}^{r,n} = t_{i,j}^{r} x_{i,j}^{n}$ and $Y_{i,j}^{k,n} = y_{i,j}^{k} x_{i,j}^{n}$. Based on McCormick envelopes, the resulting linearized formulation for \textit{SO-DDU} is: 
\begin{subequations}
\label{SO_DDU_formulation}
\begin{align}
    & \min_{t,x,w} ~ \sum_{i,j,r} h_{i,j}^{r} t_{i,j}^{r} + \sum_{n} p_{n} \bigg[ \rho \sum_{i,j} \bar{d}_{i,j} x_{i,j}^{n} + \rho \sum_{i,j} \xi_n \hat{d}_{i,j} x_{i,j}^{n}  \nonumber\\ 
    & + \sum_{i} s_{i} w_{i}^{n} - \rho \sum_{i,j,r} \xi_n \hat{d}_{i,j} \gamma_{i,j}^{r} Z_{i,j}^{r,n} + \rho \sum_{i,j} \xi_n \hat{d}_{i,j} u_{i,j} Y_{i,j}^{k,n} \bigg]  \label{so_obj}\\
    & \text{s.t.}  ~ ~ \eqref{budget},~\eqref{link_upgrade},~ \eqref{var_constr1}, \nonumber\\
    & 0 \leq \sum_{i} x_{i,j}^{n} \leq C_{j}, ~ \forall j,n,\\
    & w_i^{n} + \sum_{j} x_{i,j}^{n} \geq \lambda_i, ~ \frac{w_i^{n}}{\lambda_i} \leq \alpha_i, ~ \forall i,n,\\
    & \xi_{n} \sum_{j} \bar{d}_{i,j} x_{i,j}^{n} + \xi_{n} \sum_{j} \hat{d}_{i,j} x_{i,j}^{n} + \sum_{j,r} \xi_{n} \hat{d}_{i,j} \gamma_{i,j}^{r} Z_{i,j}^{r,n} \nonumber\\
    & + \xi_{n} \sum_{j}  \hat{d}_{i,j} u_{i,j} \sum_{k=1}^{Q_{i,j}} 2^{k-1} Y_{i,j}^{k,n} \leq \Delta_i \lambda_{i}, \forall i,n,\\
    & Z_{i,j}^{r,n} \leq L_{i,j} t_{i,j}^{r},~ Z_{i,j}^{r,n} \leq x_{i,j}^{n}, \forall i,j,r,n,\\
    & Z_{i,j}^{r,n} \geq x_{i,j}^{n} - L_{i,j} (1 - t_{i,j}^{r}), ~ \forall i,j,r,n,\\
    & Y_{i,j}^{k,n} \leq L_{i,j} y_{i,j}^{k,n}, ~ y_{i,j}^{k,n} \leq x_{i,j}^{n}, ~ \forall i,j,k,n,\\
    & Y_{i,j}^{k,n} \geq x_{i,j}^{n} - L_{i,j} (1 - y_{i,j}^{k,n}), ~ \forall i,j,k,n,\\
    & y_{i,j}^{k,n} \in \{0,1\}, ~ \forall i,j,k,n.
\end{align}
\end{subequations}

\subsection{Proof of Theorem \ref{Theo:RDDU}}
\label{proofRDDU}
The robust constraint \eqref{eq:MaxRobustConstOrg} is equivalent to:
\begin{align} \label{eq:MaxRobustConst}
    \max_{\bm{\zeta} \in \mathcal{U}(\bz)}  \bm{\zeta}^{T}\bu \leq b.
\end{align}
Based on LP duality,  \eqref{eq:MaxRobustConst} can be expressed as:
\begin{subequations}\label{lp_dual}
\begin{align}
    \bm{\pi}^{T} (\bv + \bm{\psi} \bz) \leq b, \label{lp_dual1} \\
    \bm{\pi}^{T} \bm{A} = \bu^{T}, ~~ \bm{\pi} \geq 0, \label{lp_dual2}
\end{align}
\end{subequations}
where $\bm{\pi}$ is the dual variable for constraints corresponding to the uncertainty set $\mathcal{U} (\bz)$. By expanding the variable space, the \eqref{lp_dual1} constraint can be rewritten as
\begin{align}
\sum_i\pi_i v_i + \sum_i\sum_j\psi_{i,j}y_{i,j} \leq b, ~\forall i,j, \label{lp_dual3} 
\end{align}
with $y_{i,j} = \pi_iz_j$. The standard robust approach allows us to rewrite the bilinear term using the Big-M method which can be implemented through the following linear inequalities:
\begin{subequations}
\label{McCormick}
\begin{align}
    y_{i,j} \leq \pi_i, ~ y_{i,j} \leq M z_j, ~\forall i,j \label{McCormick_ub},\\
    0 \leq y_{i,j} \geq \pi_i - M (1 - z_j), ~\forall i,j \label{McCormick_lb},
\end{align}
\end{subequations}
which completes the proof.

\subsection{Proof of Theorem \ref{Theo:e-RDDU}}
\label{proofe-RDDU}
According to Appendix~\ref{proofRDDU}, constraint  \eqref{eq:MaxRobustConstOrg} can be reformulated as \eqref{lp_dual} and \eqref{lp_dual3}. The simplification of the reformulation in this theorem compared to Theorem \ref{Theo:RDDU} is enabled by utilizing the properties of positive influence coefficients. We start the proof by rewriting constraints \eqref{lp_dual1} such that the coefficients of binary variables are positive as follows
\begin{subequations}\label{eq:Theoe-RDDUproof}
\begin{align}
    \sum_i\pi_i (v_i +\sum_{j:\psi_{i,j} < 0}\psi_{i,j}) + \sum_i\sum_{j:\psi_{i,j} \ge 0}\psi_{i,j}y_{i,j}~~~~~\cr
    - \sum_i\sum_{j:\psi_{i,j} < 0}\psi_{i,j}w_{i,j} \leq b, \label{eq:Theoe-RDDUproof_1} \\
    y_{i,j} \ge \pi_iz_j, ~\forall i,j : \psi_{i,j} \ge 0, \label{eq:Theoe-RDDUproof_2}\\
    w_{i,j} \ge \pi_i(1-z_j), ~\forall i,j : \psi_{i,j} < 0. \label{eq:Theoe-RDDUproof_3}
\end{align}
\end{subequations}

Indeed, if there exists a feasible variable for \eqref{lp_dual3} then we can find a feasible variable for \eqref{eq:Theoe-RDDUproof} by using $y_{i,j} = \pi_iz_j, ~\forall i,j : \psi_{i,j} \ge 0$ and $w_{i,j} = \pi_i(1-z_j), ~\forall i,j : \psi_{i,j} < 0$. On the other hand, if there is a feasible solution to \eqref{eq:Theoe-RDDUproof} then since the influence coefficents of the binary variables $y_{i,j}$ and $w_{i,j}$ are positive, we have
\begin{align*}
    \sum_i\pi_i (v_i +\sum_{j:\psi_{i,j} < 0}\psi_{i,j}) + \sum_i\sum_{j:\psi_{i,j} \ge 0}\psi_{i,j}\pi_iz_j&\cr
    - \sum_i\sum_{j:\psi_{i,j} < 0}\psi_{i,j}\pi_i(1-z_j) \leq \sum_i\pi_i (v_i +\sum_{j:\psi_{i,j} < 0}\psi_{i,j})& \cr
    + \sum_i\sum_{j:\psi_{i,j} \ge 0}\psi_{i,j}y_{i,j}
    - \sum_i\sum_{j:\psi_{i,j} < 0}\psi_{i,j}w_{i,j}&\leq b,
\end{align*}
where the first inequality follows from \eqref{eq:Theoe-RDDUproof_2} and \eqref{eq:Theoe-RDDUproof_3} and the last inequality follows from \eqref{eq:Theoe-RDDUproof_1},
which shows that it is also a feasible solution to \eqref{lp_dual3}. Furthermore, if $z_j=0$ then constraints \eqref{eq:Theoe-RDDUproof_2} imply that $y_{i,j}\ge 0$ and  constraints \eqref{eq:Theoe-RDDUproof_3} imply that $w_{i,j}\ge \pi_i$. If $z_j=1$ then $y_{i,j}\ge \pi_i$ and $w_{i,j}\ge 0$ which can be expressed as the constraints in Theorem \ref{Theo:e-RDDU} and the proof is completed.

\subsection{Linearization of actual workload allocation model \eqref{Actual_model}}
\label{actual_linear}
This section will provide the linearization of actual workload allocation model:
\begin{subequations}
\label{actual_linear_model}
\begin{align}
    &\min_{\begin{subarray}{c} \bw \in \mathbb{Z}_{+}^{I}, \bx \in \mathbb{Z}_{+}^{I \times J}\\ \bm{Y} \in \mathbb{Z}_{+}^{I \times J \times K} \end{subarray} } ~ \rho \sum_{i,j} \bar{d}_{i,j} x_{i,j} \!+\! \sum_{i,j} \hat{d}_{i,j} (1 - \sum_{r} \gamma_{i,j}^{r} \Tilde{t}_{i,j}^{r}) x_{i,j} \nonumber \\
    &\qquad\qquad\quad~~ + \sum_{i,j} \sum_{k=1}^{Q_{i,j}} \hat{d}_{i,j} u_{i,j} 2^{k-1} Y_{i,j}^{k} + \sum_{i} s_i w_i \\
    &  \text{s.t.} ~ \eqref{resource_cap}-\eqref{QoS}; ~~ \sum_{j}  \bar{d}_{i,j} \frac{x_{i,j}}{\lambda_i} \!+\! \sum_{j} \hat{d}_{i,j} (1 - \sum_{r} \gamma_{i,j}^{r} \Tilde{t}_{i,j}^{r}) \frac{x_{i,j}}{\lambda_i} \nonumber \\
    & + \sum_{j} \sum_{k=1}^{Q_{i,j}} \hat{d}_{i,j} 2^{k-1} \frac{Y_{i,j}^{k}}{\lambda_{i}}\leq \Delta_i, ~~\forall i, \\
    &  Y_{i,j}^{k} \leq L_{i,j} y_{i,j}^{k}; ~~  Y_{i,j}^{k} \leq x_{i,j}, ~~ \forall i,j,k, \\
    & 0 \leq Y_{i,j}^{k} \geq x_{i,j} - L_{i,j} (1 - y_{i,j}^{k}), ~~ \forall i,j,k,
\end{align}
\end{subequations}
where $L_{i,j} =  \min \{C_j, \lambda_i \}$ is an upper bound of $x_{i,j}$.

\end{document}